\documentclass[journal]{IEEEtran} 
 
\normalsize

 \usepackage{cite}
\usepackage[]{graphicx}
\graphicspath{{figures/}}
\DeclareGraphicsExtensions{.eps,.pdf}
\usepackage[mathscr]{euscript}
\usepackage{balance}

\usepackage{dsfont}
\usepackage{hyperref}
\usepackage[usenames,dvipsnames]{color}

\usepackage{graphicx}
\usepackage{subcaption} 
\usepackage{latexsym}
\usepackage{amssymb}
\usepackage{verbatim}
\usepackage{amsmath}
\usepackage{amsthm}
\usepackage{amsfonts}
\usepackage{stackengine}
\usepackage{bm} 
\usepackage{float}
\usepackage{stfloats}
\usepackage{mathtools}
\usepackage{dsfont}

\usepackage{multicol}

\newtheorem{Theorem}{Theorem}

\newtheorem{Definition}{Definition}
\newtheorem{cor}{Corollary}
\newtheorem{Remark}{Remark}

\newtheorem{Lemma}{Lemma}

\newtheorem{Example}{Example}

\IEEEoverridecommandlockouts

\begin{document}
%
\title{Coded Computation Against Processing Delays for Virtualized Cloud-Based Channel Decoding }

\author{
	\IEEEauthorblockN{{Malihe~Aliasgari}\IEEEauthorrefmark{1},~\IEEEmembership{Member,~IEEE,} {J\"org Kliewer}\IEEEauthorrefmark{1},~\IEEEmembership{Senior Member,~IEEE,} and {Osvaldo Simeone}\IEEEauthorrefmark{2}}~\IEEEmembership{Fellow,~IEEE,}\\
	\thanks{This work was supported in part by U.S. NSF grants CNS-1526547, CNS-1815322, CCF-1525629, and by the	European Research Council (ERC) under the European Union Horizon 2020 research and innovative programme (grant agreement No 725731).
		
	
	Part of the material in this paper {\color{black} was presented at IEEE International Symposium on Information Theory (ISIT), 2018 \cite{aliasgari2017codedisit}}. 
}

\thanks{M. Aliasgari and J. Kliewer are with the Department of Electrical and Computer Engineering, New Jersey Institute of Technology, Newark, New Jersey 07102-1982 USA (email: ma839@njit.edu; jkliewer@njit.edu)}
\thanks{O. Simeone is with the Department of Informatics, King's College London, London, UK (email: osvaldo.simeone@kcl.ac.uk) }
}



\maketitle

\begin{abstract}
The uplink of a cloud radio access network architecture is studied in which decoding at the cloud takes place via network function virtualization on commercial off-the-shelf servers. In order to mitigate the impact of straggling decoders in this platform, a novel coding strategy is proposed, whereby the cloud re-encodes the received frames via a linear code before distributing them to the decoding processors. Transmission of a single frame is considered first, and upper bounds on the resulting frame unavailability probability as a function of the decoding latency are derived by assuming a binary symmetric channel for uplink communications. Then, the analysis is extended to account for random frame arrival times. In this case, the trade-off between average decoding latency and the frame error rate is studied for two different queuing policies, whereby the servers carry out per-frame decoding or continuous decoding, respectively.
 Numerical examples demonstrate that the bounds are useful tools for code design and that coding is instrumental in obtaining a desirable compromise between decoding latency and reliability.
\end{abstract}

 \IEEEpeerreviewmaketitle

\section{Introduction}
Promoted  by the European Telecommunications Standards Institute (ETSI), network function virtualization (NFV) has become a cornerstone of the envisaged architecture for 5G systems \cite{mijumbi2016network}. NFV leverages virtualization technologies in order to implement network functionalities on commercial off-the-shelf (COTS) programmable hardware, such as general purpose servers, potentially reducing both capital and operating costs. An important challenge in the deployment of NFV is ensuring carrier grade performance while relying on COTS components. Such components may be subject to temporary unavailability due to malfunctioning, and are generally characterized by randomness in their execution runtimes. The typical solution to these problems involves replicating the virtual machines that execute given network functions on multiple processors, e.g., cores or servers \cite{ETSI,liu2016reliability,herrera2016resource,kang2017trade}.

{\color{black}  Among the key applications of NFV is the implementation of centralized radio access functionalities in a cloud radio access network (C-RAN) \cite{nikaein2015processing,ETSINFVCRAN}. As shown in Fig.~\ref{fignfv}, each remote radio head (RRH) of a C-RAN architecture is connected to a cloud processor by means of a fronthaul (FH) link. Baseband functionalities are carried out on a distributed computing platform in the cloud, which can be conveniently programmed and reconfigured using NFV. 
The most expensive baseband function in terms of latency to be carried out at the cloud is uplink channel decoding \cite{nikaein2015processing,alyafawi2015critical,nikaein2014openairinterface}.

 The implementation of channel decoding in the cloud by means of NFV is faced not only with the challenge of providing reliable operation despite the unreliability of COTS servers, but also with the latency constraints imposed by retransmission protocols.
 In particular, keeping decoding latency at a minimum is a major challenge in the implementation of C-RAN owing to timing constraints from the link-layer retransmission protocols \cite{dotsch2013quantitative,rost2014opportunistic,khalili2017uplink}.
 In fact, positive or negative feedback signals need to be sent to the users within a strict deadline in order to ensure the proper operation of the protocol. 
 In \cite{Rodriguez17,rodriguez2018cloud} it is argued that exploiting parallelism across multiple cores in the cloud can reduce the decoding latency by enabling decoding as soon as one can has computed its task.  
   However, parallel processing does not address the  unreliability of COTS hardware. A different solution is needed in order to address both unreliability and delays associated with cloud decoding.


  }

The problem of straggling processors, that is, of processors
lagging behind in the execution of a certain orchestrated
function, has been well studied in the context of distributed computing
{\color{black} \cite{dean,ananthanarayanan2010reining,zaharia,li2016coded,li2015coded,li2016unified}}. Recently, it has been pointed out that, for the important case
of linear functions, it is possible to improve over repetition strategies in terms of the trade-off between performance and
latency by carrying out linear precoding of the data prior to processing, {\color{black} e.g.,  \cite{Ramchandran,Li,Yang,Tandon,Dutta,Sev17,yu2017polynomial,mallick2018rateless,kosaian2018learning}}. The key idea
is that, by employing suitable linear (erasure) block codes operating over fractions of size $1/K$ of the original data, a function may be completed as soon as any  $K$ or more processors, depending on the minimum distance of the code, have completed their operations. 
Coding has also been found to be useful {\color{black}addressing} the straggler problem in the context of coded distributed storage {\color{black} and computing} systems, {\color{black} see, e.g.,  \cite{wang2015using,joshi2017efficient,ananthanarayanan2013effective,yang2018coded,aktas2017effective}.}

{ \color{black} 
In this paper, we explore the use of coded computing {\color{black} to enable} reliable and timely channel decoding in a C-RAN architecture based on distributed unreliable processors. Specifically, we formally and systematically address the analysis of coded NFV for C-RAN uplink decoding.
The only prior work on coded computing for NFV is \cite{Ali}, which provides numerical results concerning a toy example with three processors in which a processor in the cloud is either on or off.
  Unlike \cite{Ali}, in this work, we derive analytical performance bounds for a general scenario with any number of servers, random computing runtimes, and random packet arrivals. Specific novel contributions are as follows.
%
%
%
\begin{itemize}
\item 
We first consider the transmission of an isolated frame, and develop analytical upper bounds on the frame unavailability probability (FUP) as a function of the allowed decoding delay. The FUP measures the probability that a frame is correctly decoded within a tolerated delay constraint.
The FUP bounds leverage large deviation results for correlated variables \cite{Janson} and depend on the properties of both the uplink linear channel code adopted at the user and the NFV linear code applied at the cloud; 
  
\item 
As a byproduct of the analysis we introduce the dependency graph of a linear code and its chromatic number as novel
relevant parameters of a linear code beside minimum distance, blocklength, and rate; 

\item 
We extend the analysis to account for random frame arrival times, and investigate the trade-off between average decoding latency and frame error rate (FER) for two different queuing policies, whereby the servers carry out either per-frame or continuous decoding;

\item We provide extensive numerical results that demonstrate the usefulness of the derived analytical bounds in both predicting the system performance and enabling the design of NFV codes.

\end{itemize}
}


The rest of the paper is organized as follows. In Section \ref{secModel}, we present the system model focusing, as in \cite{Ali}, on a binary symmetric channel (BSC) for uplink communications. Section \ref{secASY} presents the two proposed upper bounds on the FUP as a function of latency. In Section \ref{SecQueue} we study the proposed system with random frame arrival times, and Section \ref{secnum} provides numerical results. 
\vspace{-.5ex}

\begin{figure*}[t!]
	\begin{center}
		\includegraphics[scale=0.64]{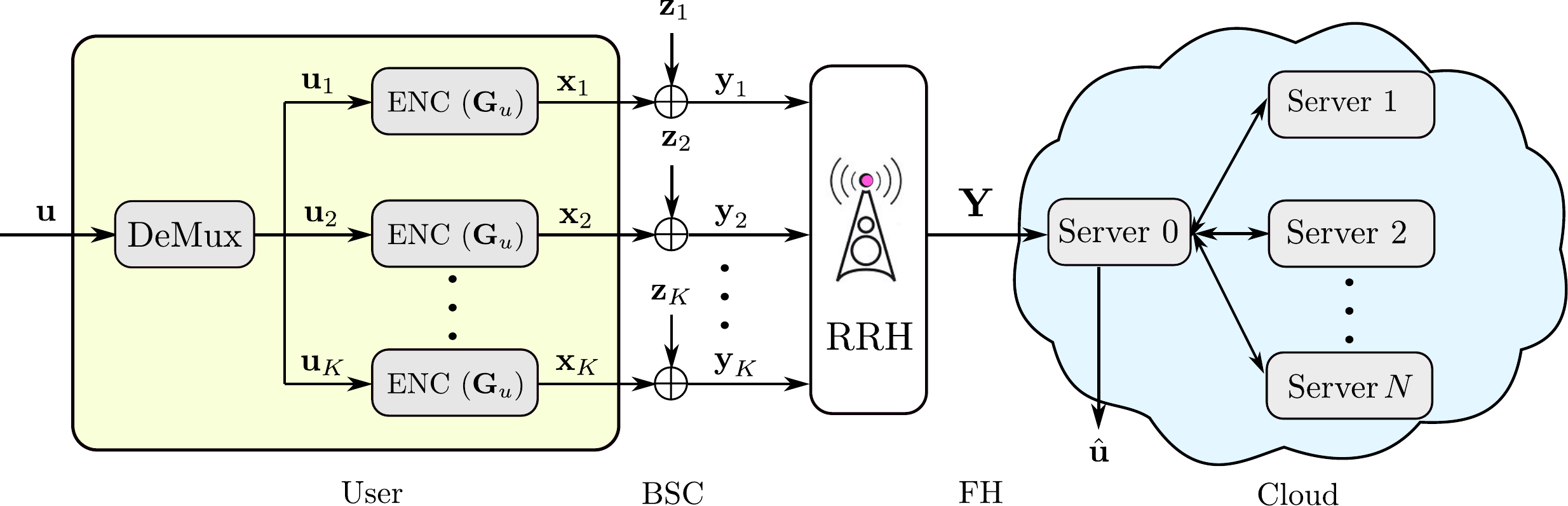}\vspace{-.1cm}~\caption{
			\footnotesize{NFV model for uplink channel
                  decoding. The input information frame $\textbf{u}$ is divided into
                  packets, which are encoded with a linear code
                  $\mathcal{C}_u$ with generator matrix
                  $\textbf{G}_u$. The packets are received by the RRH
                  through a BSC and forwarded to the cloud. Server 0 in the cloud
                  re-encodes the received packet with a linear code
                  $\mathcal{C}_c$ in order to enhance the
                  robustness against potentially straggling Servers
                  $1,\ldots,N$. 
            }}~\label{fignfv}
	\end{center}
\vspace{-6ex}
\end{figure*}
  
\section{System Model}\label{secModel}
As illustrated in Fig.~\ref{fignfv}, we consider the uplink of a C-RAN system  in which a user communicates with the cloud via a remote radio head (RRH). The user is connected to the RRH via a BSC with cross error probability $\delta$, while the RRH-to-cloud link, typically referred to as fronthaul, is assumed to be noiseless.
Note that the BSC is a simple model for the uplink channel, while the noiseless fronthaul accounts for a typical deployment with higher capacity fiber optic cables.
{\color{black} As we briefly discuss in Section \ref{secConclusion}, the analysis can be generalized to other additive noise channel, such as Gaussian channels. }
The cloud contains a master server, or Server 0, and $N$ slave servers, 
i.e., Servers $1,\ldots,N$. The slave servers are characterized by random computing delays as in related
works on  coded computation \cite{Ramchandran,Li,Sev17}. Note that we use here the term ``server" to refer to a decoding processor, although, in a practical implementation, this may correspond to a core of the cloud computing platform \cite{Rodriguez17,rodriguez2018cloud}.
\begin{figure}[t!]
	\begin{center}
		\includegraphics[scale=0.45]{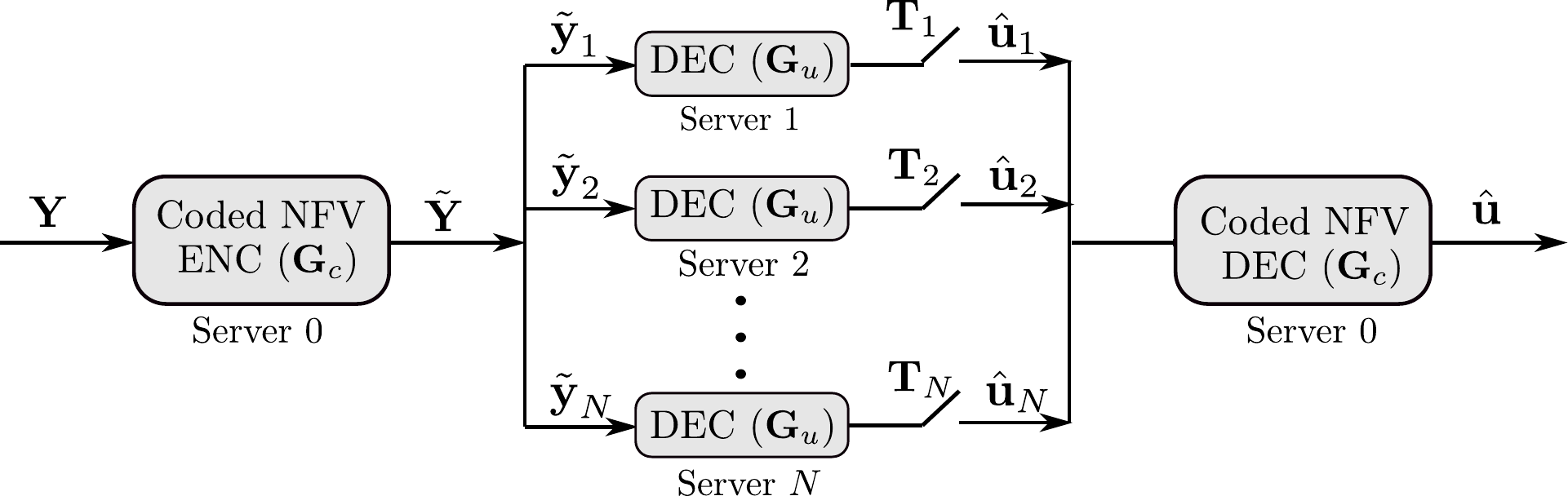}\vspace{-.1cm}~\caption{\footnotesize{
				Coded NFV at the cloud: Server 0 re-encodes the
				received packets in $\textbf{Y}$  by a linear NFV
				code $\mathcal{C}_c$ with generator  $\textbf{G}_c$.
				Each encoded packet $\tilde{\textbf{y}}_i$ is
				then conveyed to Server $i$  for decoding.
		}}~\label{figcodednfv}
	\end{center}
	\vspace{-5ex}
\end{figure}

In the first part of this paper, we consider transmission of a single information frame $\textbf{u}$, while Section \ref{SecQueue} focuses on random frame arrival times and queuing effect delays.
 The user encodes an information frame $\textbf{u}$ consisting of $L$ bits. Before encoding, the information frame is divided into $K$ blocks $\textbf{u}_1,\textbf{u}_2,\ldots , \textbf{u}_K \in \{0,1\}^{L/K}$  of equal size, each of them containing $L/K$ bits. As shown in Fig.~\ref{fignfv}, in order to combat noise on the BSC, the $L/K$ blocks are encoded by an $(n,k)$ binary linear code $\mathcal{C}_u$ of rate $r=k/n$ defined by generator matrix $\textbf{G}_u\in \mathbb{F}_2^{n\times k}$, where $n=L/(rK)$ and $k=L/K$.
Let $\textbf{x}_j\in \{0,1\}^n$ with $j\in\{1,\ldots,K\}$ be the $K$ transmitted packets of length $n$. 
At the output of the BSC, the length-$n$ received vector for the $j$th packet at the RRH is given as 
\begin{equation}\label{eqEncod}
	\textbf{y}_j=\textbf{x}_j\oplus \textbf{z}_j,
\end{equation}
 where $\textbf{z}_j$ is a vector of i.i.d. $\mathrm{Bern}(\delta$) random variables (rvs).
 The $K$ received packets $(\textbf{y}_1,\textbf{y}_2,\ldots,\textbf{y}_K)$ by the RRH are transmitted to the cloud via the fronthaul link, and the cloud performs decoding. Specifically, as detailed next, we assume that each Server $1,\ldots,N$ performs decoding of a single packet of length $n$ bits while Server 0 acts as coordinator.
 
Assuming $N\geq K$, we adopt the idea of NFV coding  proposed in \cite{Ali}. Accordingly,
as seen in Fig.~\ref{figcodednfv}, the $K$ packets are first linearly encoded by Server 0 into $N\geq K$ coded blocks of the same length $n$ bits,  each forwarded to a different server for decoding. 
This form of encoding is meant to mitigate the effect of straggling servers in a manner similar to \cite{Ramchandran,Li,Sev17}.
Using an $(N,K)$ binary linear NFV code $\mathcal{C}_c$ with $K\times N$ generator
matrix $\textbf{G}_c\in\mathbb{F}_2^{N\times K}$, the encoded packets are obtained as 
 \begin{equation}
 \tilde{\textbf{Y}}=\textbf{Y}\textbf{G}_c,
 \end{equation}
  where $\textbf{Y}=[\textbf{y}_1,\ldots,\textbf{y}_K]$  is the  $n\times K$ matrix obtained by including the received signal $\textbf{y}_j$ as the $j$th column and $\tilde{\textbf{Y}}=[\tilde{\textbf{y}}_1,\ldots, \tilde{\textbf{y}}_N]$ is
  the $n\times N$ matrix whose $i$th column $\tilde{\textbf{y}}_i$ is the input to Server $i$, where $i\in\{1,\ldots,N\}$.
 From (\ref{eqEncod}), this vector can be written as
\vspace{-.25cm}
 \begin{equation}\label{eqnoise}
 \tilde{\textbf y}_i=\sum_{j=1}^K \textbf{y}_j  g_{c,ji} =\sum _{j=1}^K \textbf{x}_j  g_{c,ji}+\sum_{j=1}^K\textbf{z}_j  g_{c,ji},
 \vspace{-.25cm}
 \end{equation}
   where $g_{c,ji}$ is the $(j,i)$th entry of matrix $\textbf{G}_{c}$.
   
 The signal part $\sum _{j=1}^K \textbf{x}_j g_{c,ji}$ in \eqref{eqnoise} is a linear combination of $d_i$ codewords for the rate-$r$ binary code with generator matrix $\textbf{G}_u$, and hence it is a codeword of the same code.
  The parameter  $d_i$, $i\in \{1,\ldots, N\}$, denotes the Hamming weight of the $i$th column  of matrix $\textbf{G}_{c}$, where $0\leq d_i\leq K$.
 Each server $i$ receives as input $\tilde{\textbf{y}}_i$ from which it can decode the codeword 
  $\sum_{i=1}^K\textbf{x}_ig_{c,ji}$. This decoding operation is affected by the noise  
   vector $\sum_{j=1}^K\textbf{z}_jg_{ji}$ in \eqref{eqnoise}, which has i.i.d. $\mathrm{Bern}(\gamma_i$) elements. Here, $\gamma_i$ is obtained as the first row and second column's entry of the matrix $\textbf{Q}^{d_i}$, with $\textbf{Q}$ being  
   {\color{black} the transition matrix of the BSC with cross over probability $\delta$, i.e.,
      \begin{equation}\label{generatorGc}
         {\textbf{Q}=
         \left[ {\begin{array}{cc}
         	 1-\delta & \delta\\
         	 \delta & 1-\delta
         \end{array} } \right].}
           \end{equation}
    As an example, $d_i=2$, implies a bit flipping probability of $\gamma_i=2\delta(1-\delta)$}. 
    Note that a larger value of $d_i$ yields a larger bit probability $\gamma_i$.
    We define as $\mathrm{P}_{n,k}(\gamma_i)$ the decoding error probability of {\color{black} the $(n,k)$ linear user code} at Server $i$, which can be upper bounded by using \cite[Theorem 33]{polyanskiy}.

  Server $i$ requires a random time $T_i=T_{1,i}+T_{2,i}$ to complete decoding, which is modeled as the sum of a component $T_{1,i}$ that is independent of the workload and a component $T_{2,i}$ that instead grows with the size $n$ of the packet processed at each server, respectively. The first component accounts, e.g., for processor unavailability periods, while the second models the execution runtime from the start of the computation. The first variable $T_{1,i}$ is assumed to have an exponential probability density function (pdf) $f_1(t)$ with mean $1/\mu_1$, while the variable $T_{2,i}$ has a shifted exponential distribution with cumulative distribution function (cdf) \cite{Amirhossein}
   \begin{equation}\label{cdfTime}
  F_2(t)=1-\exp{\left(-\frac{rK\mu_2}{L}\left(t-a\frac{L}{rK}\right)\right),}
  \end{equation}
  for $t\geq aL/(rK)$ and $F_2(t)=0$ otherwise. The parameter $a$ represents the minimum processing time per input bit, while $1/\mu_2$ is the average additional time needed to process one bit.
  {\color{black} 
  	As argued in \cite{Ramchandran,Amirhossein}, the shifted exponential model provides a good fit for the distribution of computation times over cloud computing environments such as Amazon EC2 clusters.
   }
   The cdf of the time $T_i$ can hence be written as the integral $F(t)=\int_{0}^{t}f_{1}(\tau)F_{2}(t-\tau)d\tau$.
 We also assume that the runtime rvs $\{T_i\}_{i=1}^N$ are mutually independent.
 Due to (\ref{cdfTime}), the probability that a given set of $l$ out of $N$ servers has finished decoding by time $t$ is given as
 \begin{equation}\label{eqOrderStatistic}
 a_l(t)={\color{black}{N \choose l}}F(t)^l(1-F(t))^{N-l}.
 \end{equation}
 
Let $d_{\min}$ be the minimum distance of the NFV code $\mathcal{C}_c$. Due to \eqref{eqnoise},  
Server 0 in the cloud is able to decode the message $\textbf {u}$ or equivalently the $K$ packets $\textbf {u}_j$ for $j\in\{1,\ldots,K\}$, as soon as $N-d_{\min}+1$ servers have decoded successfully. Let $\hat{\textbf{u}}_i$ be the output of the $i$th server in the cloud upon decoding. We assume that an error detection mechanism, such as a cyclic redundancy check (CRC), is in place so that Server 0 outputs
\[
\hat{\textbf{u}}_i = 
\begin{cases}
\hat{\textbf{u}}_i,& \text{for correct decoding},\\
\emptyset,              & \text{otherwise}.
\end{cases}
\]
The output $\hat{\textbf{u}}(t)$ of the decoder at Server 0 at time $t$ is then a function of the vectors $\hat{\textbf{u}}_i(t)$ for $i\in\{1,\ldots,N\}$, where 
\[
\hat{\textbf{u}}_i(t)= 
\begin{cases}
\hat{\textbf{u}}_i,& \text{if} ~T_i\leq t,\\
\emptyset,              & \text{otherwise}.
\end{cases}
\]
Finally, the frame unavailability probability (FUP) at time $t$ is defined as the probability
\begin{equation}\label{eqgoalerror}
\mathrm{P}_u(t)=\mathrm{Pr}\left[ \hat{\textbf{u}}(t)\neq\textbf{u}\right]. 
\end{equation} 
The event $\{\hat{\textbf{u}}(t)\neq\textbf{u}\}$ occurs when
either not enough servers have completed decoding or many servers have completed but failed decoding by time $t$. We also define the FER as 
\begin{equation}\label{FEReq}
\mathrm{P}_e = \lim _{t\rightarrow \infty} \mathrm{P}_u(t).
\end{equation}
The FER measures the probability that, when all servers have completed decoding, a sufficiently large number, namely larger than $N-d_{\min}$, has decoded successfully.
\section{Bounds on the Frame Unavailability Probability}\label{secASY}
In this section we derive analytical bounds on the FUP $\mathrm{P}_u(t)$ in \eqref{eqgoalerror} as a function of the decoding latency $t$.
\subsection{Preliminaries}
Each server $i$ with $i\in\{1,\ldots,N\}$ decodes successfully its assigned packet $\tilde{\textbf{y}}_i$ if: (\textit{i}) the server completes decoding by time $t$;
(\textit{ii}) the decoder at the server is able to correct the errors caused by the BSC. 
Furthermore as discussed, an error at Server 0 occurs at time $t$ if the number of servers that have successfully decoded by time $t$ is smaller than $N-d_{\min}+1$.

To evaluate the FUP, we hence define the indicator variables
$C_i(t)=\mathds{1}\{T_i\leq t\}$ and $D_i$
which are equal to 1 if the events (\textit{i}) and (\textit{ii}) described above occur, respectively, and zero otherwise. Based on these definitions, the FUP is equal to 
\begin{eqnarray}\label{eqFER}
\mathrm{P}_u(t)&=& \mathrm{Pr}\left[ \sum_{i=1}^N C_i(t) D_i\leq N-d_{\min}  \right].
\end{eqnarray}
The indicator variables $C_i(t)$ are independent Bernoulli rvs across the servers $i\in\{1,\ldots,N\}$, due to the independence assumption on the rvs $T_i$. However, the indicator variable $D_i$ are dependent Bernoulli rvs.
The dependence of the variables $D_i$ is caused by the fact that the noise terms 
$\sum_{i=1}^K\textbf{z}_jg_{c,ji}$ in \eqref{eqnoise} generally have common terms. 
In particular, if two columns $i$ and $j$ of the generator matrix $\textbf{G}_c$ have at least a 1 in the same row, then the decoding indicators $D_i$ and $D_j$ are correlated.
 This complicates the evaluation of bounds on the FUP \eqref{eqFER}.

\subsection{ Dependency Graph and Chromatic Number  of a Linear Code}\label{subsecdependency}
  To capture the correlation among the indicator variables $D_i$, we
  introduce here the notion of the  \emph{dependency graph} and its
  chromatic number for a linear code. These appear to be novel properties of a linear code, and we will argue below that they determine the performance of the NFV code $\mathcal{C}_c$ for the application at hand.
\begin{Definition}\label{dependencyGraph}
	Let $\textbf{G}\in \mathbb{F}_2^{K'\times N'}$ be a generator matrix of a linear code. The dependency graph $\mathcal{G}(\textbf{G})=(\mathcal{V},\mathcal{E})$ comprises a set $\mathcal{V}$ of $N'$ vertices and a set $\mathcal{E}\subseteq \mathcal{V}\times \mathcal{V}$ of edges, where edge $(i,j)\in \mathcal{E}$ is included if both the $i$th and $j$th columns of $\textbf{G}$ have at least a $1$ in the same row.
\end{Definition}

\begin{Example}\label{exGraph}
For an $(8,4)$ NFV code $\mathcal{C}_c$ with the following generator matrix  
   \begin{equation}\label{generatorGc}
   { 
    \textbf{G}_c=
   \left[ {\begin{array}{cccccccc}
   	1&0&0&0&0&1&1&0\\
   	0&0&0&1&1&0&0&1\\
   	0&1&0&0&0&0&1&1\\
   	1&0&1&0&1&0&0&0
   	\end{array} } \right],}
     \end{equation}
   the resulting dependency graph $\mathcal{G}(\textbf{G}_c)$ is shown in Fig.~\ref{fig:graph}.
\end{Example}
 \begin{figure}[t!]
 \centering
 \includegraphics[scale=0.42]{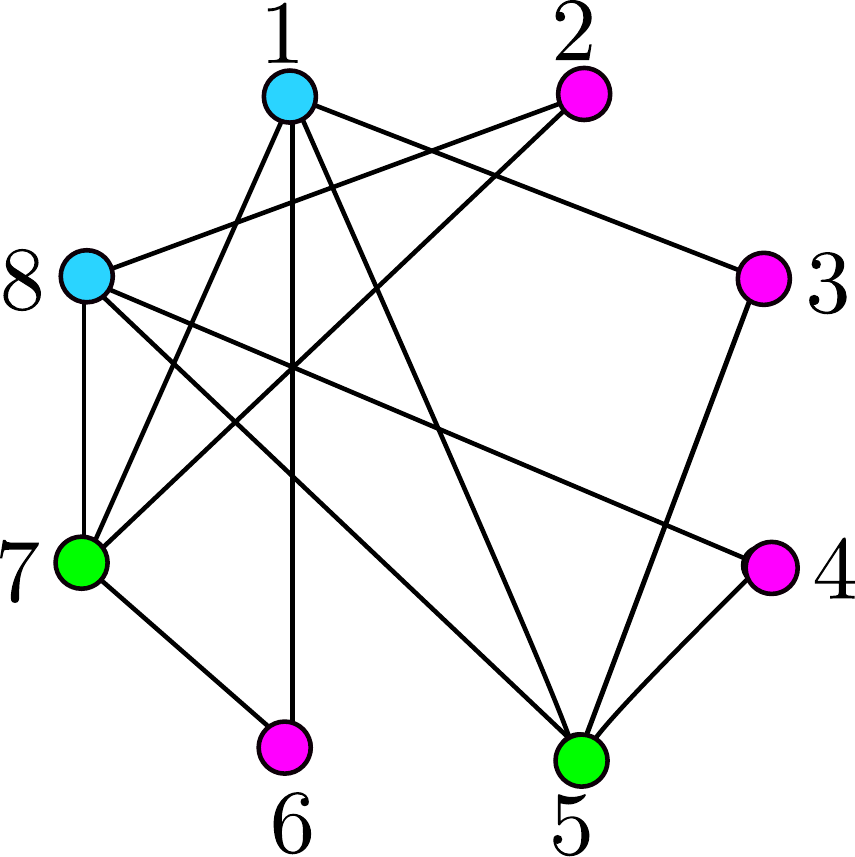} ~\caption{\footnotesize{Dependency graph associated with the (8,4) NFV code $\mathcal{C}_c$ in Example \ref{exGraph}.}}~\label{fig:graph}
 \vspace{-5ex}
 \end{figure}
	\begin{figure*}[!b]
	{\fontsize{10pt}{12pt} 
		\hrulefill	
		\begin{equation}\label{eqasymm}
		 		\begin{split}
		 		\mathrm{P}_u(t)\leq 
		 		\exp\left(-\frac{S(t)}{b^2(t)\mathcal{X}(\textbf{G}_c)}~\varphi\left(\frac{4b(t)\left(NF(t)-F(t)\sum_{i=1}^N\mathrm{P}_{n,k}(\gamma_i)-N+d_{\min}\right)}{5S(t)}\right) \right),
		 		\end{split}
		 		\end{equation}  }
\end{figure*}
  		
The chromatic number $\mathcal{X}(\textbf{G})$ of the graph $\mathcal{G}(\textbf{G})$ will play an important role in the analysis. We recall that the chromatic number is 
 the smallest number of colors needed to color the vertices of  $\mathcal{G}(\textbf{G})$, such that no two adjacent vertices share the same color (see the example in Fig.~\ref{fig:graph}).   
Generally, finding the chromatic number of a graph is NP-hard \cite{NP}. However, a simple upper bound on $\mathcal{X}(\textbf{G})$ is given as \cite{Brook}
\begin{equation}\label{Brook}
\mathcal{X}(\textbf{G})\leq \Delta(\textbf{G})+1,
\end{equation}
where $\Delta(\textbf{G})$ is the maximum degree of a graph $\mathcal{G}(\textbf{G})$.
 A consequence of \eqref{Brook} is the following.
 \begin{Lemma}
 	Let $\textbf{G} $ be a $K'\times N'$ matrix, where $\alpha_r$ and $\alpha_c$ are the maximum Hamming weights of the rows and columns in $\textbf{G}$, respectively. Then the chromatic number of the corresponding dependency graph $\mathcal{G}(\textbf{G})$ is upper bounded as 
 	\begin{equation}\label{eqChromatic}
 	\mathcal{X}(\textbf{G})\leq \min\{N,\alpha_c(\alpha_r-1)+1 \}.
	\end{equation}
\end{Lemma}

\begin{proof}
According to Definition \ref{dependencyGraph}  we have the upper bound 	
$\Delta(\textbf{G})\leq \alpha_c (\alpha_r-1)$ and hence \eqref{eqChromatic} follows directly from \eqref{Brook}.
\end{proof}
 
\vspace{-.5ex}
\subsection{Large Deviation Upper Bound}\label{subseclargedevi}
	\begin{figure*}[!b]
	{\fontsize{10pt}{12pt} 
		\begin{equation}\label{uppereqTigther}
		\begin{split}
		\mathrm{P}_u(t)\leq 1- \frac{1}{{N \choose l}} \sum_{l=N-d_{\min}+1}^N   a_l(t) \sum_{\substack{ \mathcal{A}\subseteq       \{1,\ldots,N\}:  \\ |\mathcal{A}|=l}} \left(1-\exp   \left(-\frac{S_{\mathcal{A}}}{b_{\mathcal{A}}^2\mathcal{X}(\textbf{G}_{\mathcal{A}})}\varphi  \left(\frac{4b_{\mathcal{A}} \left (l-N+d_{\min}-\mathrm{P}_{n,k}^{\mathcal{A}} \right )}{5S_{\mathcal{A}}}\right)\right )  \right).
		\end{split}
		\end{equation}  }
\end{figure*}
In this subsection, we {\color{black} derive an upper bound} on the FUP. The bound is based on the large deviation result in \cite{Janson} for
the tail probabilities of rvs
 $X=\sum_{i=1}^M X_i$, where the rvs $X_{i}$ are generally dependent. {\color{black} We refer to this bound as the large deviation bound (LDB)}.
 The correlation of rvs $\{X_i\}$ is described in \cite{Janson} by a dependency graph. This is defined as any graph $\mathcal{G}(X)$ with $X_i$ as vertices, such that, if a vertex $i \in \{1,\ldots,M\}\backslash\{i\}$ is not connected to any vertex in a subset $\mathcal{J}\subset\{1,\ldots,M\}$, then $X_i$ is independent of $\{X_j\}_{j\in\mathcal{J} }$.
 
\begin{Lemma}[\!\cite{Janson}]\label{JansonLemmaTighter} 
	Let $X=\sum _{i=1}^M{X_i}$, where $X_{i}\sim \text{Bern} (p_i)$ and $p_i\in(0,1)$ are generally dependent.
	For any $b\geq 0$, such that the inequality
		$X_{i} -\mathbb{E}(X_{i})\geq -b$
	holds for all $i\in\{1,\ldots,M\}$ with probability one, and for any
	$\tau\geq 0$
	we have
	\begin{equation}\label{eqlem1}
	\mathrm{Pr}[X\leq \mathbb{E}(X)-\tau]\leq \exp\left(-\frac{S}{b^2\mathcal{X}(\mathcal{G}(X))}~\varphi{\left (\frac{4b\tau}{5S}\right )} \right),
	\end{equation}
	where
	$S\overset{\Delta}{=}\sum_{i=1}^N\text{Var}(X_i)$
	and
	$\varphi(x)\overset{\Delta}{=}(1+x)\ln(1+x)-x$.
	The same bound \eqref{eqlem1} holds for 
	$\mathrm{Pr}(X\geq \mathbb{E}(X)+\tau)$, where
	$X_i -\mathbb{E}(X_i)\leq b$
 with probability one.
\end{Lemma}
The following theorem uses Lemma \ref{JansonLemmaTighter} to derive a bound on the FUP.
\begin{Theorem}\label{ThmDevi}
Let $\mathrm{P}_{n,k}^{\min}=\min_i\{\mathrm{P}_{n,k}(\gamma_i)\}_{i=1}^N$.
	For all
	\begin{equation}\label{conditionTimemu10}
	t\geq F^{-1}\left ( \frac{N-d_{\min} }{N-\sum_{i=1}^N\mathrm{P}_{n,k}(\gamma_i)} \right ),
	\end{equation}
 the FUP is upper bounded by
 in {\color{black}\eqref{eqasymm}}, shown at the bottom of the page,	
	where
	$b(t)\overset{\Delta}{=} F(t)\left (1- \mathrm{P}_{n,k}^{\min}\right ) $
	and
	 $S(t)\overset{\Delta}{=} \sum_{i=1}^N F(t)\left (1-\mathrm{P}_{n,k}(\gamma_i)\right ) \left (1-F(t)(1-\mathrm{P}_{n,k}(\gamma_i)) \right )$.
The upper bound \eqref{eqasymm} on the FUP captures the dependency of the FUP on both the channel and the NFV code. In particular, the bound is an increasing function of the error probabilities $\mathrm{P}_{n,k}(\gamma_i)$, which depend on both codes. It also depends on the NFV code through parameters $d_{\min}$ and $\mathcal{X}(\textbf{G}_c)$.
\end{Theorem}

\begin{proof}
	Let  $X_i(t)\overset{\Delta}{=}C_i(t)D_i$ and $X(t)=\sum_{i=1}^NX_i(t)$, where 
$X_i(t)$ are dependent Bernoulli rvs with probability $\mathbb{E}[X_i(t)]=\mathrm{Pr}[X_i(t)=1]=F(t)\left (1-\mathrm{P}_{n,k}(\gamma_i)\right )$.
It can be seen that a valid dependency  graph $\mathcal{G}(X)$ for the variables $\{X_i\}$ is the dependency graph
$\mathcal{G}(\textbf{G}_c)$ defined above. This is due to the fact that, as
discussed in Section~\ref{subseclargedevi}, the rvs~$X_i$ and $X_j$ are dependent if and only if the $i$th and $j$th  column of $\textbf{G}_c$ have at least a 1 in a common row.
 We can hence apply Lemma \ref{JansonLemmaTighter} for every time $t$ by selecting 
	$\tau=\mathbb{E}(X )-N+d_{\min} $,  and $b(t)$ as defined
        above. Note that this choice of
        $b(t)$ meets the
        constraint for $b$ in Lemma \ref{JansonLemmaTighter}.
        For $1/ \mu_1=0$, \eqref{conditionTimemu10} can be simplified as follows:
        \vspace{-.3cm}
                \begin{equation}\label{conditionTime}
                t\geq n\left (a-\frac{1}{\mu}\ln \left( \frac{d_{\min}-\sum_{i=1}^N\mathrm{P}_{n,k}(\gamma_i)}{N-\sum_{i=1}^N\mathrm{P}_{n,k}(\gamma_i)} \right ) \right ).\vspace{-.3cm}
                \end{equation}
 \end{proof}
\vspace{-1ex}
{\color{black}
\begin{Remark}
	When $t\rightarrow \infty$,  we have the limit $\lim_{t\rightarrow\infty} F(t)=1$, which implies that eventually all servers complete decoding. Letting $d^{\max}\overset{\Delta}{=}\max\{d_i\}_{i=1}^N$ and 
	$\gamma\overset{\Delta}{=}\textbf{Q}^{d^{\max}}(1,2)$,  the first row and second column's entry of the matrix $\textbf{Q}^{d^{\max}}$, the bound \eqref{eqasymm} reduces to
	\begin{align}\label{eqineq1}
	\!{\color{black}\lim_{t\rightarrow\infty}\!\mathrm{P}_u(t)}\!\leq\!\exp \!\Bigg(\!\frac{-N\mathrm{P}_{n,k}(\gamma)}{(1\!-\!\mathrm{P}_{n,k}(\!\gamma)\!)\mathcal{X}(\!\textbf{G}_c)\!}
		 \varphi\Bigg(\hspace{-1ex}\frac{4\!\left(\!d_{\min}/N\!-\!\mathrm{P}_{n,k}(\!\gamma)\!\right)}{5\mathrm{P}_{n,k}(\gamma)}\hspace{-1ex}\Bigg)\hspace{-1ex}\Bigg).
		\end{align}
	This expression demonstrates the dependence of the FUP bound \eqref{eqasymm} on the number of servers $N$, the decoding error probability  $\mathrm{P}_{n,k}(\gamma)$ for each server, the chromatic number $\mathcal{X}(\textbf{G}_c)$, and minimum distance $d_{\min}$ of the NFV code.
	In particular, it can be seen that the FUP upper bound \eqref{eqineq1} is a decreasing function of  $d_{\min}$, while it increases with the chromatic number, $\mathrm{P}_{n,k}(\gamma)$ and with $d^{\max}$. 
	
\end{Remark}
}
\vspace{-2ex}
\subsection{Union Bound}\label{subsecunion} 
 As indicated in Theorem \ref{ThmDevi},
 the large deviation based bound in \eqref{uppereqTigther} is only valid for large
 enough $t$, as can be observed from  \eqref{conditionTime}.
 Furthermore, it may generally not be tight, since it neglects the independence of the indicator variables $C_i$.
  In this subsection, a generally tighter but more complex {\color{black} union bound (UB)} is derived that is valid for all times $t$.
 \begin{Theorem}\label{ErrorBoundThm}
 	For any subset $\mathcal{A}\subseteq\{1,\ldots,N\}$, define 
  \begin{equation*}
   \mathrm{P}_{n,k}^{\min(\mathcal{A})}\overset{\Delta}{=}\min\{\mathrm{P}_{n,k}(\gamma_i)\}_{i\in\mathcal{A}}~~ \text{and}~~
  \mathrm{P}_{n,k}^{\mathcal{A}} \overset{\Delta}{=}\sum_{i\in \mathcal{A}} \mathrm{P}_{n,k}(\gamma_i),
  \end{equation*}
       and let
   $\textbf{G}_{\mathcal{A}}$ 
   be the
   $K\times |\mathcal{A}|$,
    submatrix of
    $\textbf{G}_c$,
     with column indices in the subset $\mathcal{A}$. 
  	Then, the FUP is upper bounded by 
  	\eqref{uppereqTigther}, shown at the bottom of the  page,
where $S_{\mathcal{A}}\triangleq\sum_{i\in \mathcal{A}} \mathrm{P}_{n,k}(\gamma_i) \left (1-\mathrm{P}_{n,k}(\gamma_i) \right )$ and $b_{\mathcal{A}}\overset{\Delta}{=}1-\mathrm{P}_{n,k}^{\min(\mathcal{A})}$.
\end{Theorem}
\begin{proof}
	Let $I_i=1-D_i$ be the indicator variable which equals 1 if Server $i$ fails decoding.
Accordingly, we have $I_i\sim \mathrm{Bern}(\mathrm{P}_{n,k}(\gamma_i))$.
	  For each subset $\mathcal{A}\subseteq \{1,\ldots ,N\}$, let 
	  $I_{\mathcal{A}}=\sum _{i\in \mathcal{A}} I_i$.
The complement of the FUP $\mathrm{P}_{s}(t)=1-\mathrm{P}_{u}(t)$ can hence be written as
\begin{align}	
\mathrm{P}_{s}(t)=& \mathrm{Pr}\left[ \sum_{i=1}^N C_i(t) D_i  > N-d_{\min}  \right]\\
 =&{\color{black}\frac{1}{{N \choose l}}\!\sum_{l=N-d_{\min}+1}^N  \!a_l(t) \hspace{-2ex}\sum_{\substack{ \mathcal{A}\subseteq \{1,\ldots,N\}:\\  |\mathcal{A}|=l}}}\nonumber\\
 	&{\color{black}\cdot \sum_{j=N-d_{\min}+1}^l{ \mathrm{Pr}\left[\substack{j~\text{servers from~}\mathcal{A}~\text{decode successfully}\\\text{ and}\\l-j~\text{servers from~}\mathcal{A}~\text{fail to decode }}\right]}}	\\
=& {\color{black}\!\frac{1}{{N \choose l}}}\sum_{l=\!N\!-\!d_{\min}\!+\!1}^N  \hspace{-2ex} a_l(t)\hspace{-2ex} \sum_{\substack{ \mathcal{A}\subseteq \{1,\ldots,N\}:  \\ |\mathcal{A}|=l}} \hspace{-3ex}\left(1\!-\!\mathrm{Pr}\!\left[I_\mathcal{A}\!\geq\! l\!-\!N\!+\!d_{min}\!\right]\right)\!.\label{eqthm}\vspace{-.3cm}
\end{align}	
We can now apply Lemma \ref{JansonLemmaTighter} to the probability in \eqref{eqthm} by noting that
$\mathcal{G}(\mathbf{G}_{\mathcal{A}})$ is a valid dependency graph for the variables $\{I_i\}$, $i\in \mathcal{A}$.
In particular, we apply Lemma \ref{JansonLemmaTighter} by setting
	 $\tau_{\mathcal{A}}= l-N+d_{\min}-\mathbb{E}(I_{\mathcal{A}}) $,
	  $b_{\mathcal{A}}\geq I_i-\mathbb{E}[I_i]$,
and
$S_\mathcal{A}=\sum_{i \in \mathcal{A}} \mathrm{Var}~ (I_i)$, leading to \vspace{-0.3cm}
   \begin{multline}\label{eqineq}
	\mathrm{Pr}\left[I_\mathcal{A}\geq l-N+d_{\min}\right]\leq\\
	\exp   \left(-\frac{S_{\mathcal{A}}}{b_{\mathcal{A}}^2\mathcal{X}(\textbf{G}_{\mathcal{A}})}\varphi  \left(\frac{4b_{\mathcal{A}} \left (l-N+d_{\min}-\mathrm{P}_{n,k}^{\mathcal{A}} \right )}{5S_{\mathcal{A}}}\right)\right ).
	\end{multline}
By substituting \eqref{eqineq} into \eqref{eqthm}, the proof is completed.
\end{proof}	
\vspace{-1.5ex}

\section{Random Arrivals and Queuing}\label{SecQueue}
In this section we extend our analysis from one to multiple frames transmitted by the users. To this end, we study the system illustrated in Fig.~\ref{fignfvqueue} with random frame arrival times and queueing at the servers. We specifically focus on the analysis of the trade-off between average latency and FER. 
\subsection{System Model}
\begin{figure*}[t!]
  	\begin{center}
  		\includegraphics[scale=0.64]{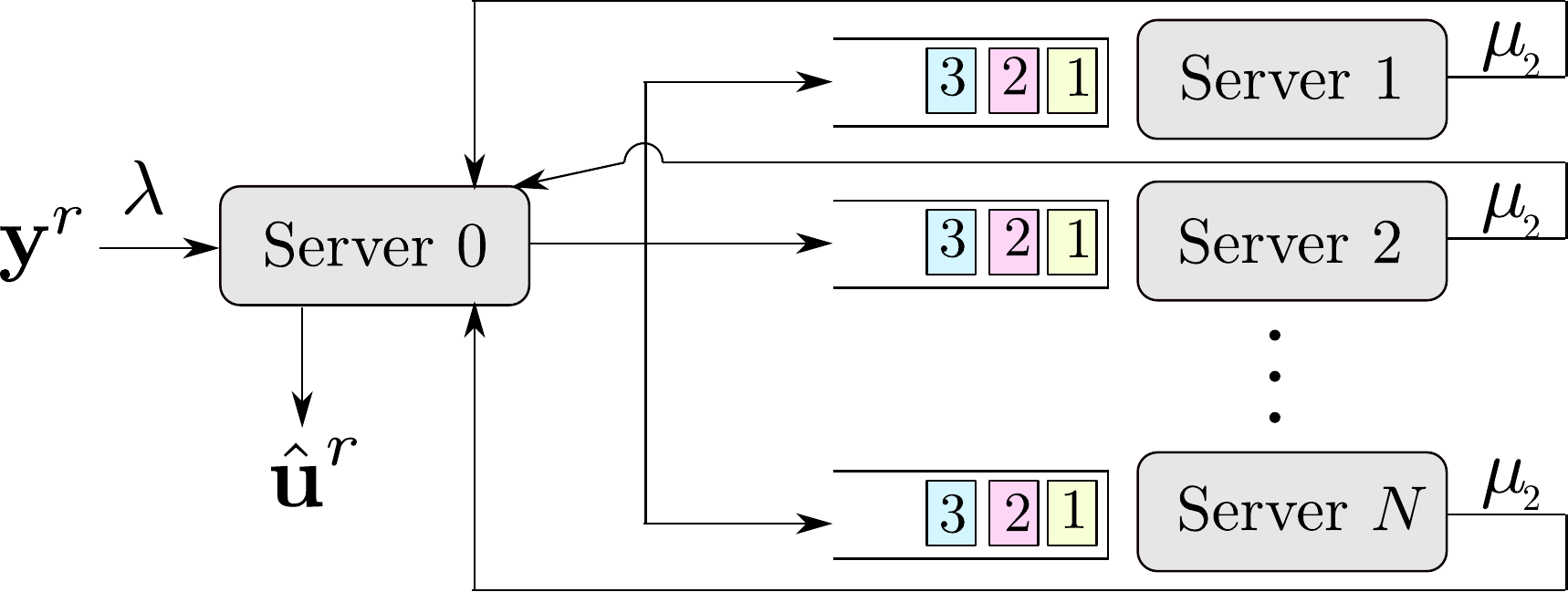}\vspace{-.1cm}~\caption{\footnotesize{In the model studied in Section \ref{SecQueue},
  				frames arrive at the receiver according to a Poisson process with parameter $\lambda$. Server 0 in the cloud encodes the received frames using an NFV code and forwards the encoded packets to servers $1,\ldots,N$ for decoding.
  		}}~\label{fignfvqueue}
  	\end{center}
  	\vspace{-5ex}
  \end{figure*}
  As illustrated in Fig.~\ref{fignfvqueue}, we assume that the arrival times of the received frames are random and distributed according to a Poisson process with a rate of $\lambda$ frames per second. Upon arrival, Server 0 applies an  NFV code to any received frame $\textbf{y}^r$ for $r=1,2,\ldots$, as described in Section II and sends each resulting coded packet $\tilde{\textbf{y}}_i^r$ to Server $i$, for $i=1,\ldots,N$. At Server $i$, each packet $\tilde{\textbf{y}}_i^r$ enters a first-come-first-serve queue. After arriving at the head of the queue, each packet $\tilde{\textbf{y}}_i^r$ requires a random time $T_i$ to be decoded by Server $i$. Here, we assume that $T_i$ is distributed according to an exponential distribution in \eqref{cdfTime} with {\color{black}  an average processing time of $1/\mu_2$ per bit.  Furthermore, the average time to process a frame of $n$ bits is denoted as $1/\mu$.}
   Also, the random variables $T_i$ are i.i.d. across servers.

 If the NFV code has minimum distance $d_{\min}$, as soon as $N-d_{\min}+1$ servers decode successfully 
 their respective packets derived from frame $\textbf{y}^r$, the information frame $\textbf{u}^r$ can be decoded at Server 0. We denote as $T$ the average overall latency for decoding frame $\textbf{u}^r$, which includes both queuing and processing.
 
Using \eqref{FEReq}, \eqref{eqFER} and the fact that all servers complete decoding almost surely as $t\rightarrow\infty$, that is $C_i(t)\rightarrow 1$ as $t\rightarrow\infty$, the FER 
 is equal to
\begin{align}\label{queueFERR}
\mathrm{P}_e =  \mathrm{Pr} \left[ \sum _{i=1}^N I_i \geq d_{\min} \right],
\end{align}
where $I_i$ is the indicator variable that equals $1$ if decoding at Server $i$ fails. This probability can be upper bounded by the following corollary of Theorem \ref{ThmDevi}.


\begin{cor}\label{propFERqueue}
 The FER defined in \eqref{queueFERR} is upper bounded by
\begin{equation}\label{eqqueueFER}
\mathrm{P}_e\leq  \exp\!\left(\!\frac{-S}{b^2\mathcal{X}(\textbf{G}_{\mathcal{C}})} \varphi\!\left(\!\frac{4b\! \left(\!d_{\min}\!-\!\sum _{i=1}^N\! \mathrm{P}_{n,k}(\gamma_i)\! \right)}{5S}\!\right)\!\right ),
\end{equation}
where $S\triangleq\sum_{i=1}^N \mathrm{P}_{n,k}(\gamma_i) \left (1-\mathrm{P}_{n,k}(\gamma_i) \right )$ and $b\overset{\Delta}{=}1-\mathrm{P}_{n,k}^{\min}$.

\begin{proof}
The result follows from Theorem \ref{ThmDevi} by selecting $\tau = d_{\min}- \sum _{i=1}^N \mathrm{P}_{n,k}(\gamma_i)$.
\end{proof}
\end{cor}

 We now discuss the computation of the average delay $T$ for different queueing management policies.

\begin{figure*}[t!]
	\centering
	\begin{subfigure}{.5\textwidth}
		\centering
		\includegraphics[scale=0.5]{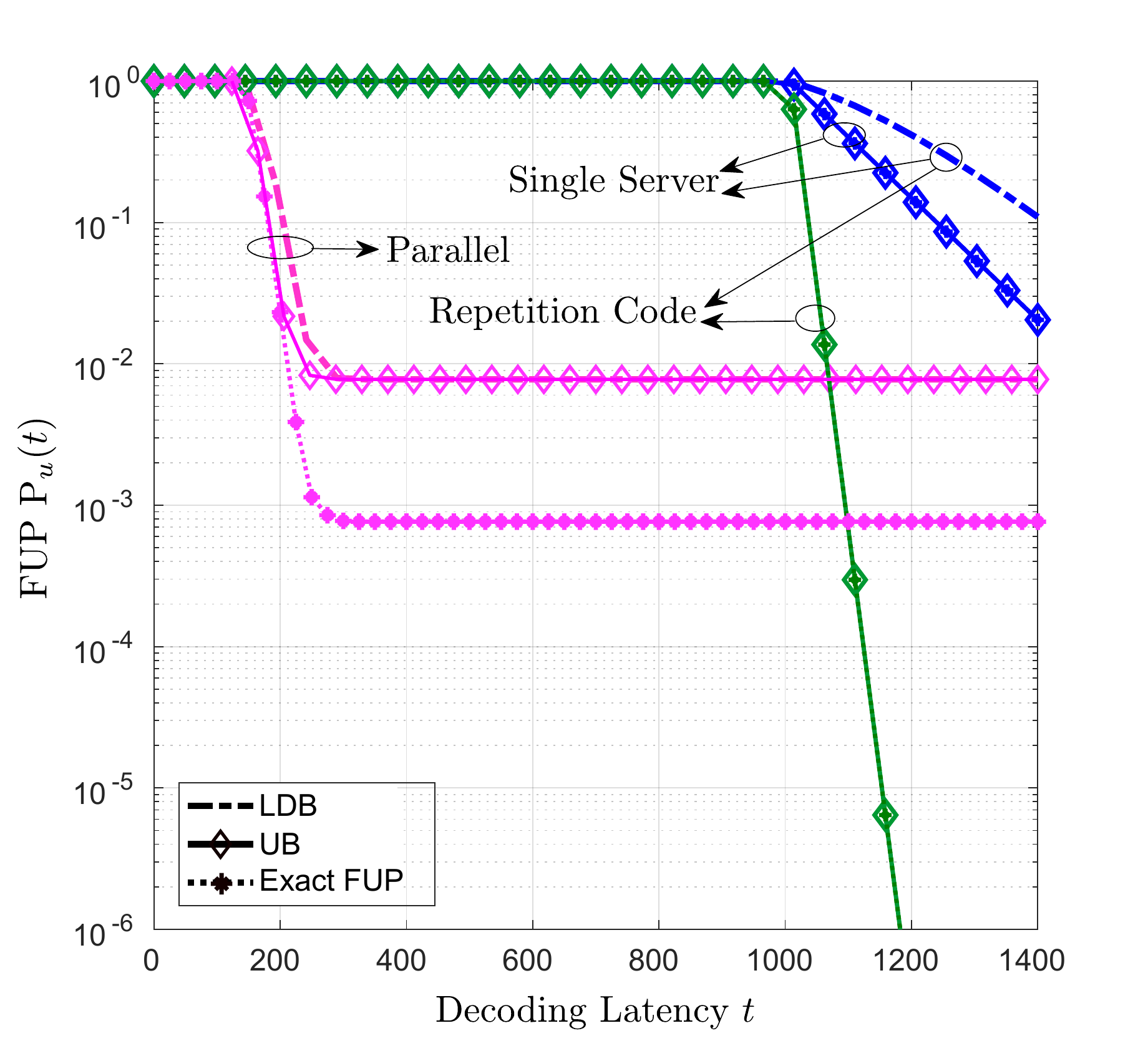}
		\caption{\footnotesize{Parallel, single server and repetition code.}}
		\label{fig:01}
	\end{subfigure}%
	\begin{subfigure}{.5\textwidth}
		\centering
		\includegraphics[scale=0.5]{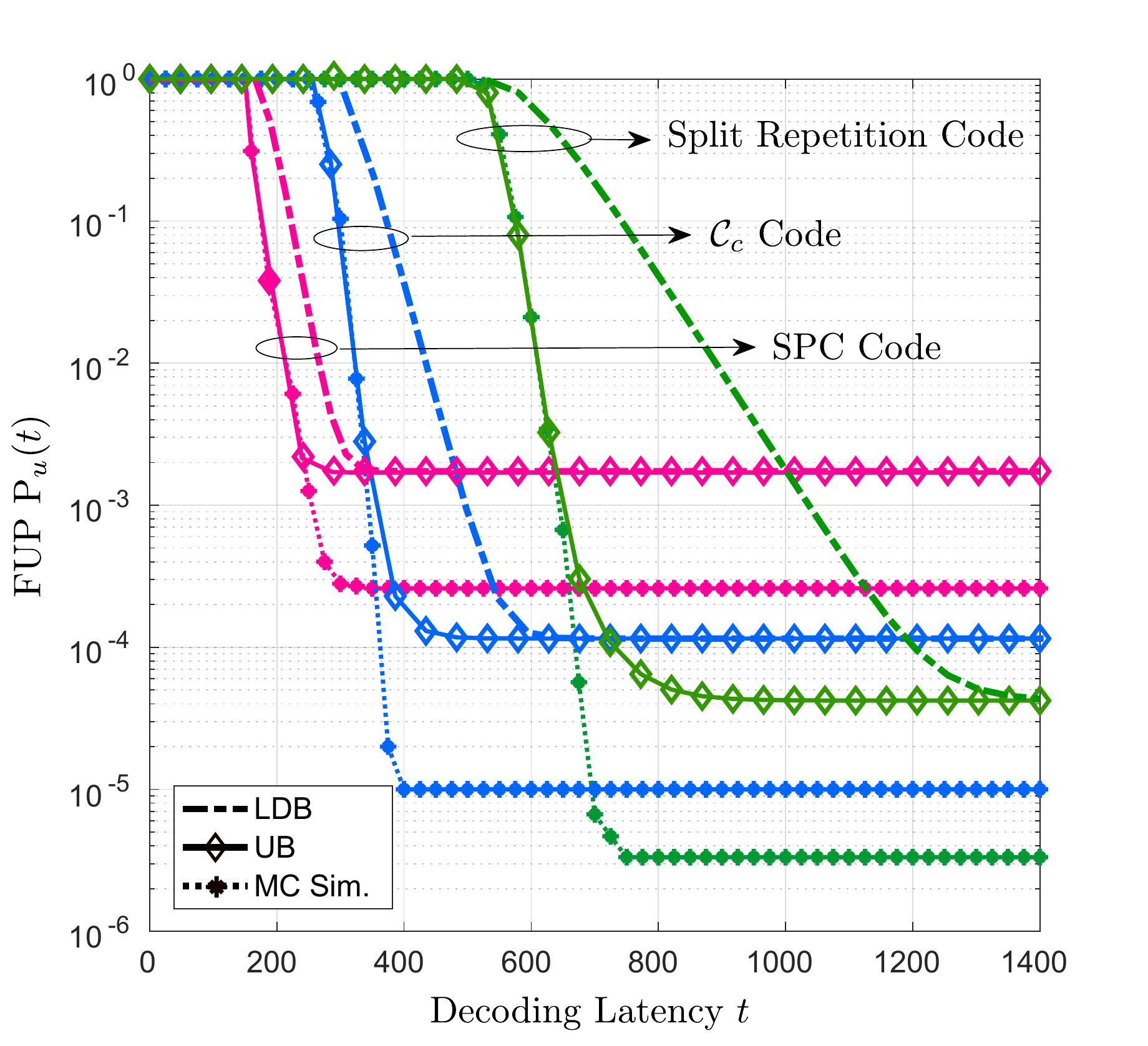}
		\caption{\footnotesize{Split repetition code, SPC code and $\mathcal{C}_c$ code.}}
		\label{fig:02}
	\end{subfigure}%
	\caption{\footnotesize{Decoding latency versus FUP for $L=504,N=8,1/\mu_1=0, \mu_2=10,a=1,\delta=0.01,r=0.5):$ (a) LDB, UB and Exact FUP for the {\color{black} parallel}, single-server, and repetition coding; (b) LDB, UB and Monte Carlo simulation (``MC Sim.'') results for split repetition code, SPC code, and the NFV code $\mathcal{C}_c$ defined in \eqref{generatorGc}.} }\label{fig:5ab}
\end{figure*}

\subsection{Per-Frame Decoding}
We first study the system under a queue management policy whereby only one frame $\textbf{y}^r$ is decoded at any time. Therefore, all servers wait until at least $N-d_{\min}+1$ servers have completed decoding of their respective packets $\tilde{\textbf{y}}^r_i$ before moving to the next frame $r+1$, if this is currently available in the queues. Furthermore, as soon as Server 0 decodes a frame, the corresponding packets still being present in the servers' queues are evicted.  

As a result, the overall system can be described 
an M/G/1 queue with arrival time $\lambda$ and service time distributed according to the $(N-d_{\min}+1)$th order statistic of the exponential distribution 
{\color{black}  \cite{Joshi}. }
The latter has the pdf \cite{ross2014introduction}
\eqref{pdfOrderStatistics}, shown at the bottom of the page,
 where $F_{T}(t)$  and $f_{T}(t)$ are the cdf and pdf of rv $T_i$, respectively. This queueing system was also studied in the context of distributed storage systems.

 Using the Pollaczek-Khinchin formula \cite{PZK}, the average delay of an M/G/1 queue can be obtained as 
 \eqref{ResponseTimeUB},  shown at the bottom of the page,
 where $H_N$ and $H_{N^2}$ are generalized harmonic numbers, defined by
$H_N=\sum_{i=1}^N\frac{1}{i}$ and $H_{N^2}=\sum_{i=1}^N\frac{1}{i^2}$ \cite{Joshi}.
Note that the queue is stable, and hence the average delay \eqref{ResponseTimeUB} is finite, if the inequality $n\lambda(H_N-H_{d_{\min}-1})<\mu (N-d_{\min}+1)$ holds. We refer to the described queue management scheme as per-frame decoding (pfd). {\color{black}{ This set-up is equivalent to the fork-join system studied in \cite{Joshi}. }}

\begin{figure*}[!b]
	{\fontsize{10pt}{12pt} 
		\hrulefill	
		\begin{equation}\label{pdfOrderStatistics}
		f_{T_{N-d_{\min}+1:N}}(t)=\frac{N!}{(N-d_{\min})!(d_{\min}-l)!}f_T(t)F_T(t)^{N-d_{\min}}(1-F_T(t))^{d_{\min}-1},
		\end{equation}
		\begin{equation}\label{ResponseTimeUB}
		T_{\text{pfd}} = \frac{n(H_N-H_{d_{\min}-1})}{(N-d_{\min}+1)\mu}+\frac{\lambda n^2 [ (H_N-H_{d_{\min}-1})^2+(H_{N^2}-H_{(d_{\min}-1)^2})]}{2 (N-d_{\min}+1)^2 \mu ^{2}  [ 1-\lambda n \mu ^{-1}(N-d_{\min}+1)^{-1} (H_N-H_{d_{\min}-1})]},
		\end{equation}}
\end{figure*}

\subsection{Continuous Decoding}
As an alternative queue management policy, as soon as any Server $i$ decodes its packet $\tilde{\textbf{y}}_i^r$, it starts decoding the next packet $\tilde{\textbf{y}}_i^{r+1}$ in its queue, if this is currently available. Furthermore, as above, as soon as Server 0 decodes a frame $\textbf{y}^r$, all corresponding packets $\tilde{\textbf{y}}^r_i$ still in the servers' queues are evicted. 
We refer this queue management policy as continuous decoding (cd).

The average delay \eqref{ResponseTimeUB} of per-frame decoding is an upper bound for the average delay of continuous decoding, i.e., we have $T_{\text{cd}} \leq T_{\text{pfd}}$ \cite{Joshi}. This is because, with per-frame decoding, all $N$ servers are blocked until $N-d_{\min}+1$ servers decode their designed packets. We evaluate the performance of continuous decoding using Monte Carlo methods in the next section.


\section{ Simulation Results}\label{secnum}

\vspace{-.5ex}
In this section we provide numerical results to provide additional insights into the performance trade-off for the system {\color{black}shown} in Fig.~\ref{fignfv}. We first consider individual frame transmission as studied in Section \ref{secModel} and Section \ref{secASY}, and then we study random arrivals as investigated in Section \ref{SecQueue}.
\subsection{Single Frame Transmission}\label{simSingle}
We first consider single frame transmission. The main  goals are to validate the usefulness of the two bounds
presented in Theorems 1 and 2 as design tools and to assess the importance of coding in obtaining desirable trade-offs between decoding latency and FUP.
We employ a frame length of $L=504$ and $N=8$ servers. The user code $\mathcal{C}_u$ is selected to be a randomly designed $(3,6)$ regular (Gallager-type) LDPC code with $r=0.5$, which is decoded via belief propagation.

We compare the performance of the following solutions: (\emph{i}) \textit{Standard single-server decoding}, whereby we assume, as a benchmark, the use of a single server, that is $N=1$, that decodes the entire frame ($K=1$); (\emph{ii}) \textit{Repetition coding}, whereby the entire frame ($K=1$) is replicated at all servers; (\emph{iii}) \textit{Parallel processing}, whereby the frame is divided into $K=N$ disjoint parts processed by different servers; 
{\color{black}(\emph{iv}) \textit{Split repetition coding}, whereby the frame is split into two parts, which are each replicated at $N/2$ servers. 
The code has hence $K=2$, $d_{\min}=N/2$, $\mathcal{X}(\textbf{G}_c)=N/2$, 
which can be thought of as an intermediate choice between repetition coding and the {\color{black} parallel} scheme;}
 (\emph{v}) \textit{Single parity check code (SPC)}, {\color{black} with $N=K+1$, whereby, in addition to the servers used by parallel decoding, an additional server decodes the} binary sum of all other $K$ received packets; and (\emph{vi}) an NFV code $\mathcal{C}_c$ with the generator matrix $\textbf{G}_c$ defined in \eqref{generatorGc},  which is characterized by $K=4$.
Note that, with both single-server decoding and repetition coding, we have a blocklength of $n=1008$ for the channel code. Single-server decoding is trivially characterized by $\mathcal{X}(\mathbf{G}_c)=d_{\min}=1$, while repetition coding is such that the equalities
$\mathcal{X}(\mathbf{G}_c)=d_{\min}=8$ hold. Furthermore, the {\color{black} parallel} approach is characterized by $n=126$, $d_\text{min}=1$ and $\mathcal{X}(\textbf{G}_c)=1$;
{\color{black} the split repetition code is characterized by $n=504 ,d_{\min}=4$ and $\mathcal{X}(\textbf{G}_c)=4$;  }
 the SPC code has $n=144, d_{\min}=2$ and $\mathcal{X}(\textbf{G}_c)=2$; and the NFV code $\mathcal{C}_c$ has $n=252$, $d_\text{min}=3$ and $\mathcal{X}(\textbf{G}_c)=3$. The exact FUP for a given function $\mathrm{P}_{n,k}(\cdot)$ can easily be computed for cases (\emph{i})-(\emph{iii}).
In particular, for single server decoding, the FUP equals
		\begin{equation}\label{ExactSingle}
		\mathrm{P}_u(t)=1-a_1(t)(1-\mathrm{P}_{{L/r},L}(\delta));
		\end{equation} 
		for the repetition code, the FUP is
		\begin{equation}\label{ExactRepetition}
		\mathrm{P}_u(t)=1-\sum_{i=1}^Na_i(t)(1-\mathrm{P}_{L/r,L}(\delta));
		\end{equation} 
		and for the {\color{black} parallel} approach, we have
		\begin{equation}\label{ExactUncoded}
		\mathrm{P}_u(t)=1-a_N(t)(1-\mathrm{P}_{L/(rN),L/N}(\delta))^N.
		\end{equation} 
In contrast, the exact FUPs for the SPC and code $\mathcal{C}_c$ are difficult to compute, due to the discussed correlation among the decoding outcomes at the servers. 

\begin{figure*}[t!]
	\centering
	\begin{subfigure}{.5\textwidth}
		\centering
		\includegraphics[scale=0.5]{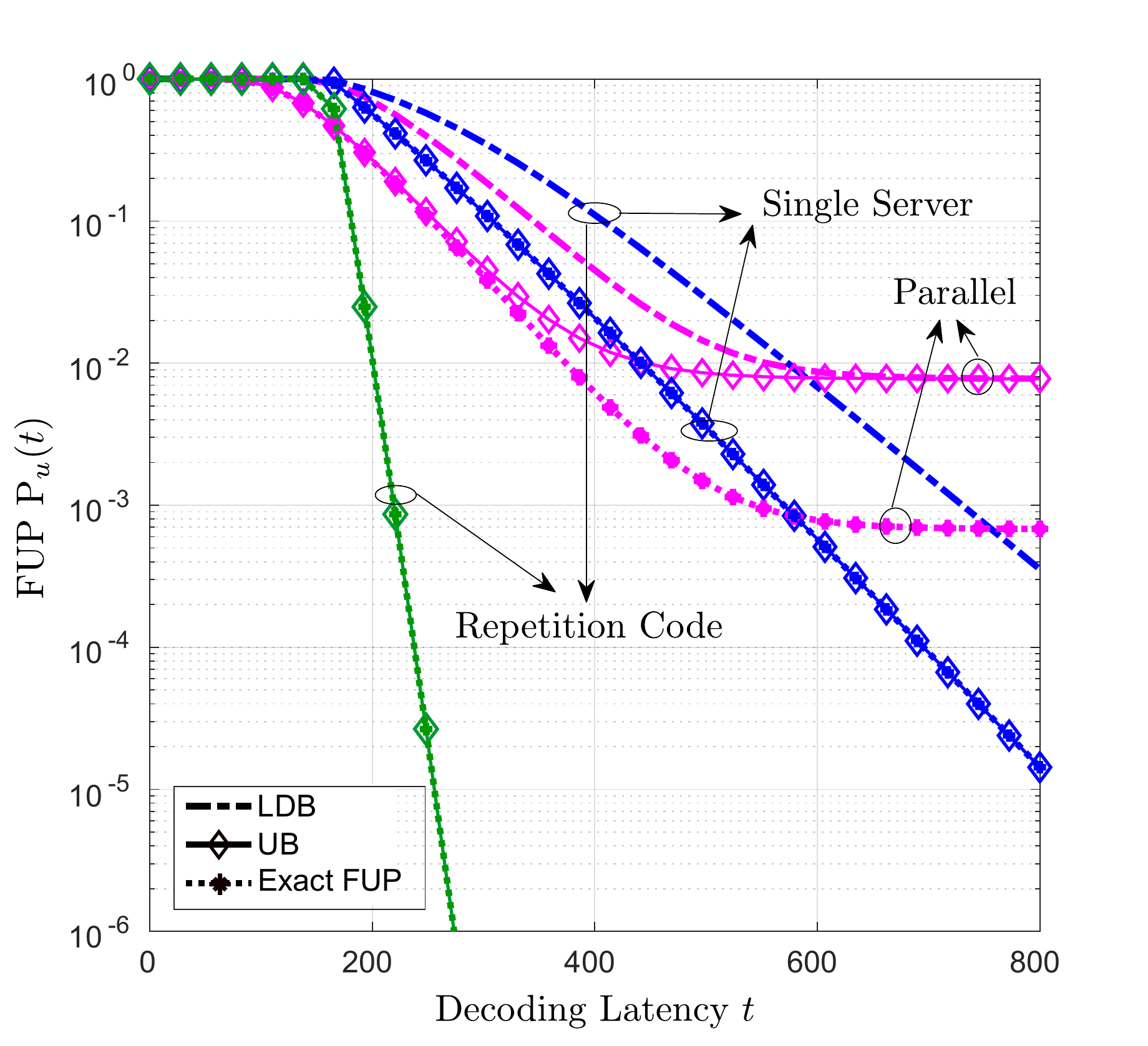}
		\caption{\footnotesize{Parallel, single server and repetition code.}}
		\label{fig:03}
	\end{subfigure}%
	\begin{subfigure}{.5\textwidth}
		\centering
		\includegraphics[scale=0.5]{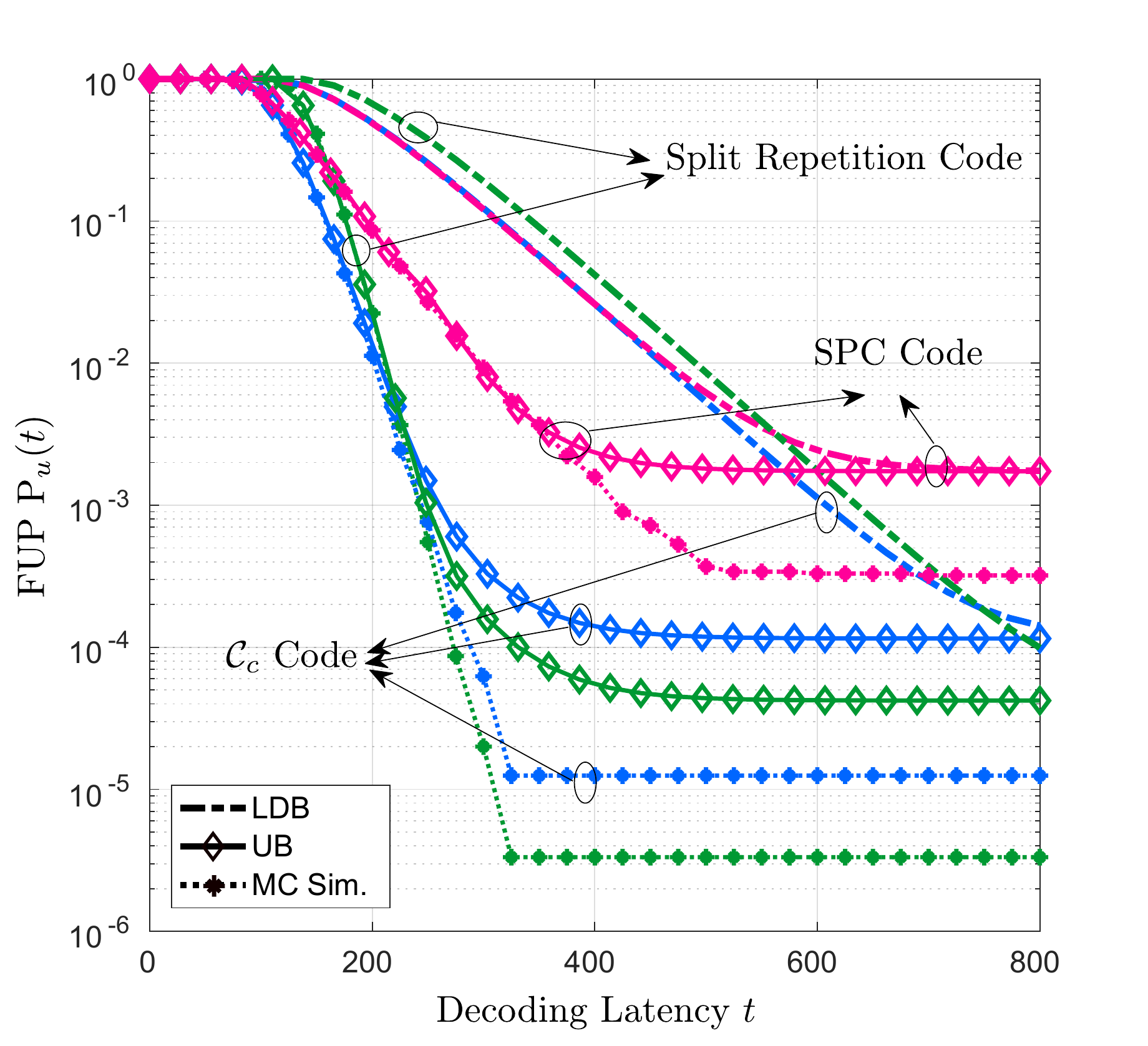}
		\caption{\footnotesize{Split repetition code, SPC code and $\mathcal{C}_c$ code.}}
		\label{fig:04}
	\end{subfigure}%
\caption{\footnotesize{Decoding latency versus FUP for $(L=504,N=8,1/\mu_1=50, \mu_2=20,a=0.1,\delta=0.01,r=0.5):$ (a) LDB, UB and Exact FUP for the {\color{black} parallel}, single-server, and repetition coding; (b) LDB, UB and Monte Carlo simulation (``MC Sim.'') results for split repetition code, SPC code, and the NFV code $\mathcal{C}_c$ defined in \eqref{generatorGc}. }}\label{fig:6ab}
\end{figure*}
 
Fig.~\ref{fig:01} shows decoding latency versus FUP for the LDB in Theorem \ref{ThmDevi}, the UB in Theorem \ref{ErrorBoundThm},
{\color{black}  and the exact error \eqref{ExactSingle}, \eqref{ExactRepetition}, \eqref{ExactUncoded}, for the first three schemes (\emph{i})-(\emph{iii}), and Fig.~\ref{fig:02} shows the LDB in Theorem \ref{ThmDevi}, the UB in Theorem \ref{ErrorBoundThm}, as well as Monte Carlo simulation results for schemes  (\emph{iv}), (\emph{v}), and (\emph{vi}). Here, we assume that the latency contribution that, is independent of the workload, is negligible, i.e., $1/\mu_1=0$. }
We also set $a=1$ and $\mu_2=10$. As a first observation, Fig.~\ref{fig:5ab} confirms that the UB bound is tighter than the LDB.

 Leveraging multiple servers in parallel for decoding is seen to yield significant gains in terms of the trade-off between latency and FUP as argued also in \cite{Rodriguez17} by using experimental results.
In particular, the {\color{black} parallel} scheme is observed to be preferred for lower latencies. This is due to the shorter blocklength $n$, which entails a smaller average decoding latency. However, the error floor of the {\color{black} parallel} scheme is large due to the higher error probability for short blocklengths. In this case, other forms of NFV coding are beneficial. To elaborate, repetition coding requires a larger latency in order to obtain acceptable FUP performance owing to the larger blocklength $n$, but it achieves a significantly lower error floor. For intermediate latencies, the SPC code, and at larger latencies also both the NFV code $\mathcal{C}_c$,
{\color{black} and the split repetition code} provide a lower FUP. This demonstrates the effectiveness of NFV encoding in obtaining a desirable trade-off between latency and FUP. 

In order to validate the conclusion obtained using the bounds, Fig.~\ref{fig:5ab} also shows the exact FUP for the schemes (\emph{i})-(\emph{iii}), as well as Monte Carlo simulation results for schemes (\emph{iv})-(\emph{vi}), respectively. While the absolute numerical values of the bounds in Fig.~\ref{fig:01} and \ref{fig:02} are not uniformly tight with respect to the actual performance, the relative performance of the coding schemes are well matched by the analytical bounds. This provides evidence of the usefulness of the derived bounds as a tool for code design in NFV systems.

 Fig.~\ref{fig:6ab} is obtained in the same way as Fig.~\ref{fig:5ab}, except for the parameters $\mu_1=0.02$, $\mu_2=20$, and $a=0.1$. 
 {\color{black} Unlike Fig.~\ref{fig:5ab}, here latency may be dominated by effects that are independent of the blocklength $n$ since we have $1/\mu_1>0$.}
 The key difference with respect to Fig.~\ref{fig:5ab} is that, for this choice of parameters,
 	repetition coding tends to outperform both the {\color{black} parallel} case, and the
 	NFV code $\mathcal{C}_c$, apart from very small latencies. 
 	This is because repetition coding has the maximum resilience to the unavailability of the servers, while not being excessively penalized by the larger blocklength $n$. This is not the case, however, for very small latency levels, where the NFV code $\mathcal{C}_c$ provides the smallest
 	FUP given its shorter blocklength as compared to repetition
 	coding and its larger $d_{\min}$, with respect to the {\color{black} parallel} scheme. 
   \begin{figure}[t!]
    	\begin{center}
    	\includegraphics[scale=0.49]{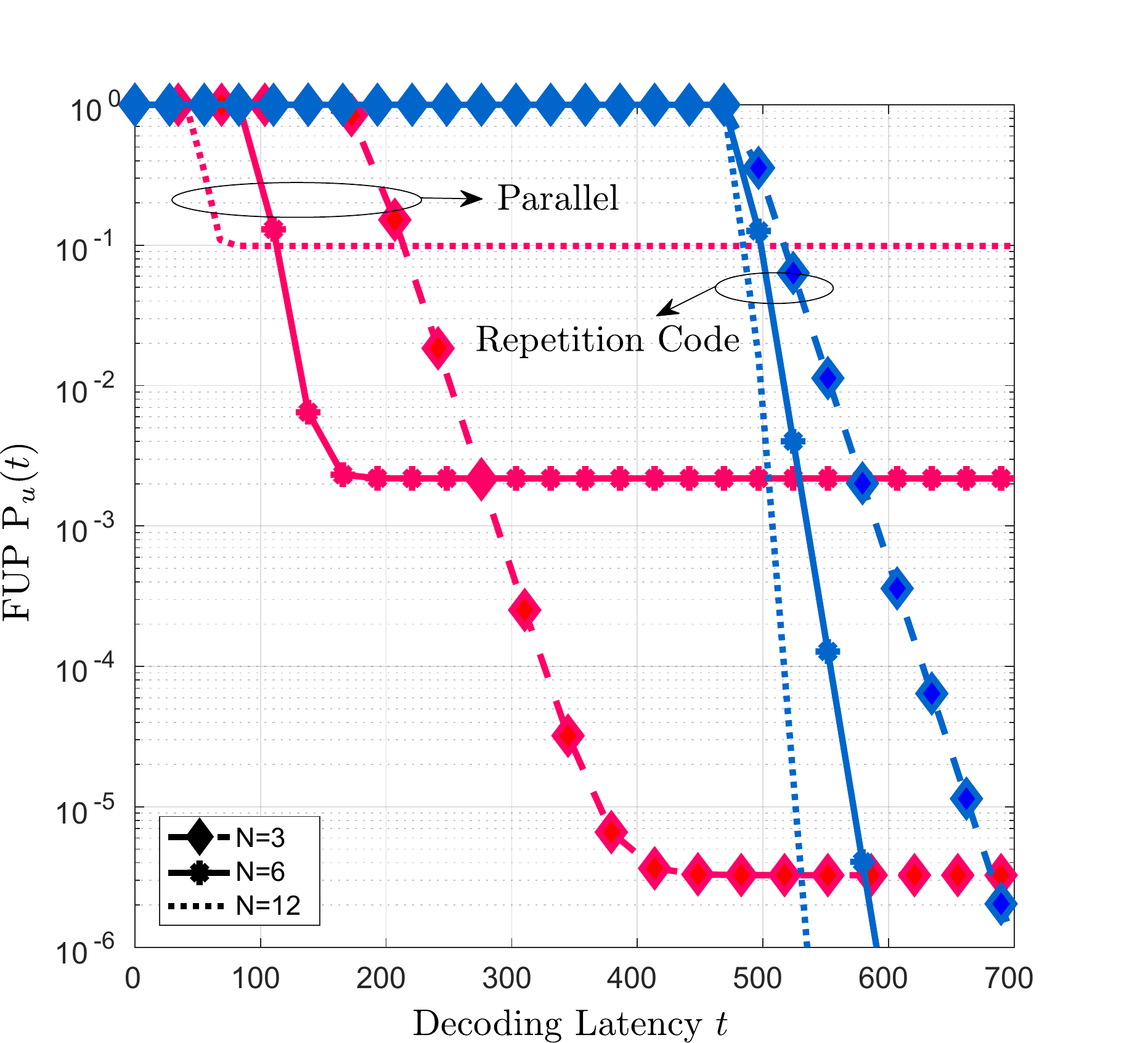} ~\caption{\footnotesize{Decoding latency versus exact FUP for parallel and repetition coding for different number of servers $N\in \{3,6,12\}$ and $(L=240,1/\mu_1=0, \mu_2=10,a=1,\delta=0.03,r=0.5)$ }
    	}~\label{figNserver}
    	\end{center}
    \end{figure}


{\color{black} Fig.~\ref{figNserver} shows the exact FUP for the extreme cases of parallel and repetition coding for different number of servers $N\in\{3,6,12\}$. The figure confirms that, for both schemes, the latency decreases for a larger number of  servers $N$. However, by increasing $N$, the error floor of the parallel scheme grows due to the higher channel error probability for shorter block lengths.
 }

\begin{figure*}[t!]
	\centering
	\begin{subfigure}{.5\textwidth}
		\centering
		\includegraphics[scale=0.5]{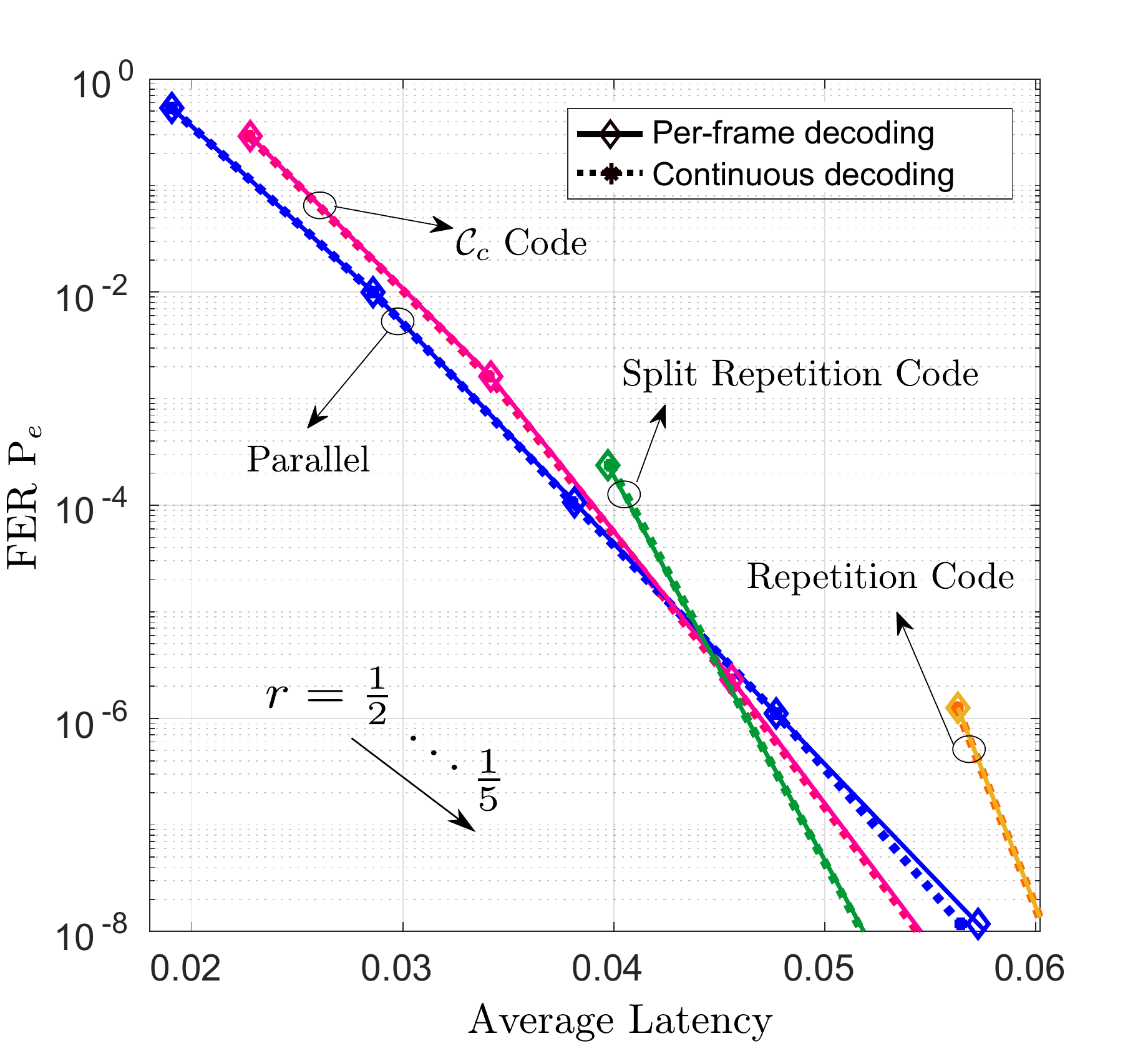}
		\caption{\footnotesize{Lightly loaded system, $\lambda=0.1, \mu =500$. }}
		\label{fig1queue}
	\end{subfigure}%
	\begin{subfigure}{.5\textwidth}
		\centering
		\includegraphics[scale=0.5]{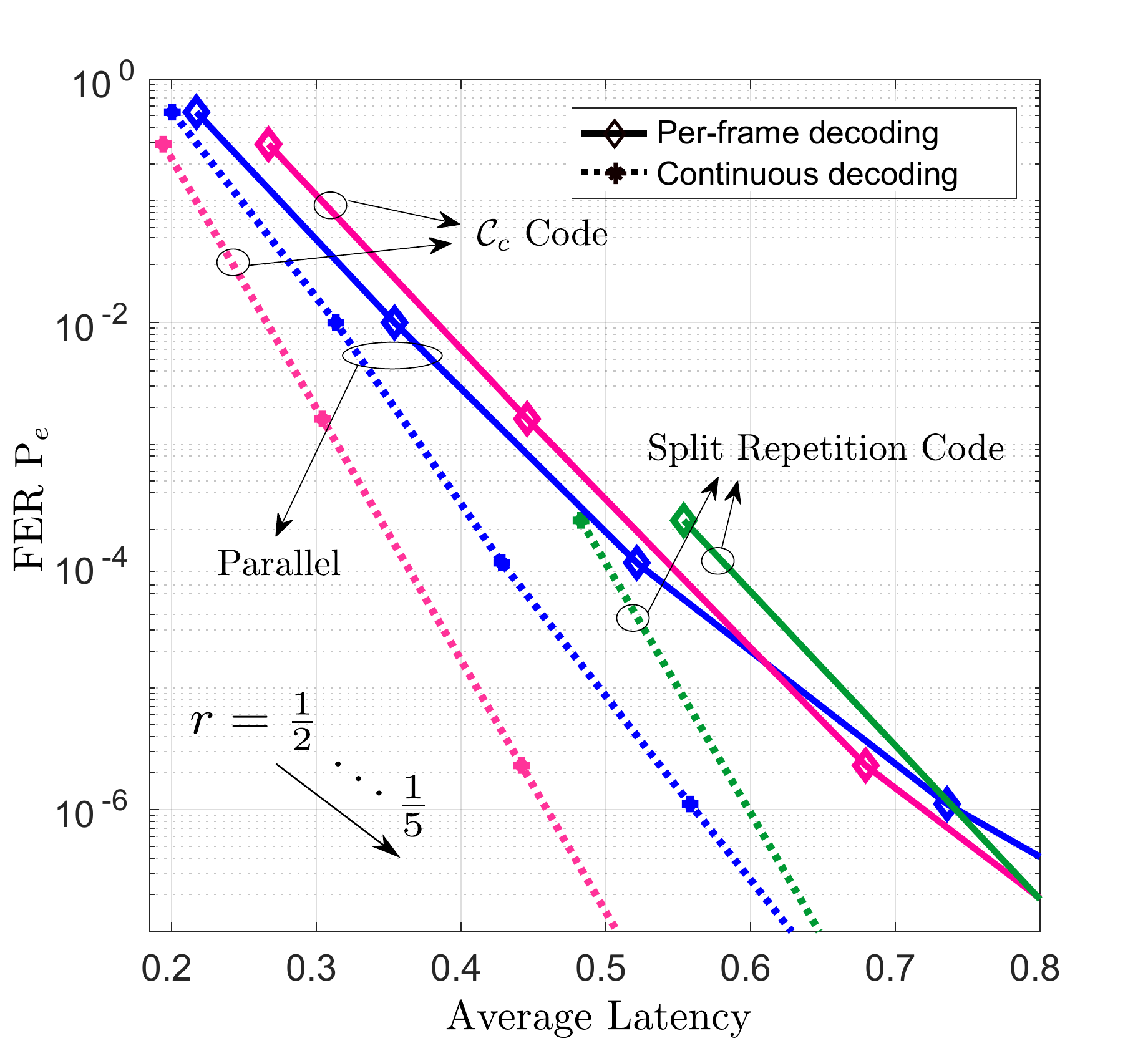}
		\caption{\footnotesize{Heavily loaded system, $\lambda=1, \mu =50$. }}
		\label{fig2queue}
	\end{subfigure}%
	\caption{ \footnotesize{Average latency versus FER with different values of the user code rate $r$ and for different coding schemes when the system is (a) lightly loaded and (b) heavily loaded, respectively ($L=112, N=8, \delta=0.03 $). }}
\end{figure*}
 
\subsection{Random Frame Transmission}\label{simQueue}
We now consider the queueing system described in Section \ref{SecQueue}, and present numerical results that provide insights into the performance of both per-frame and continuous decoding in terms of FER versus average latency \eqref{queueFERR}. {\color{black} As above, the decoding error probability is upper bounded by using \cite[Theorem 33]{polyanskiy}.} Both FER and average latency are a function of the user code rate $r$. We hence vary $r\in\{1/2,\ldots,1/5\}$ to parametrize a trade-off curve between FER and latency. We assume a frame length of $L=112$ bits with $N=8$  servers, and adopt the same user code $\mathcal{C}_c$ as in the previous subsection. The average delay $T_{\text{pfd}}$ is computed from \eqref{ResponseTimeUB}, and $T_{\text{cd}}$ is obtained via Monte Carlo simulations. 

Figs. \ref{fig1queue} and \ref{fig2queue} compare the performance of repetition coding, the NFV code $\mathcal{C}_c$ with the generator matrix \eqref{generatorGc}, and the {\color{black} parallel} approach as defined above.   
Fig.~\ref{fig1queue} considers a lightly loaded system with $\lambda=0.1$ frames per second and  $\mu  = 500$ frames per second, while Fig.~\ref{fig2queue} shows a highly loaded system with both $\lambda=1$ frames per second and $\mu = 50$ frames per second.

First, by comparing the two figures we observe that  per-frame decoding and continuous decoding have a similar performance when the system is lightly loaded (see Fig.~\ref{fig1queue}), while continuous decoding yields a smaller average latency than per-frame decoding when the system is heavily loaded (see Fig.~\ref{fig2queue}). This is because, in the former case, it is likely that a frame is decoded successfully  before the next one arrives. This is in contrast to heavily loaded systems in which the average latency becomes dominated by queuing delays.
We also note that, for repetition coding, the performance of per-frame decoding and continuous decoding coincides in both lightly or heavily loaded systems, {\color{black} since } decoding is complete as soon as one server decodes successfully.

{\color{black} Also, by} comparing the performance of different codes, we recover some of the main insights obtained from the study of the isolated frame transmission. In particular, the {\color{black} parallel} approach outperforms all other schemes for low average delays due to its shorter block length $n$. In contrast, repetition coding outperforms all other schemes in FER for large average delay because of its large block length $n$ and consequently low probability of decoding error (not shown). Furthermore, we observe that split repetition coding is to be preferred for small values of FER.

{\color{black} Finally, Fig.~\ref{fig111queue} demonstrates the behavior of the average latency as the arrival rate $\lambda$ increases and the system becomes more heavily loaded. 
We observe that, for a lightly loaded system, 
the latencies of per frame and continuous decoding are similar, while continuous decoding is preferable for a large number of $\lambda$. This is because
per-frame decoding requires
all servers to wait until at least $N-d_{\min}+1$ servers have completed decoding of their respective packets before moving on to the next frame.}

  \begin{figure}[t!]
   	\begin{center}
   	\includegraphics[scale=0.5]{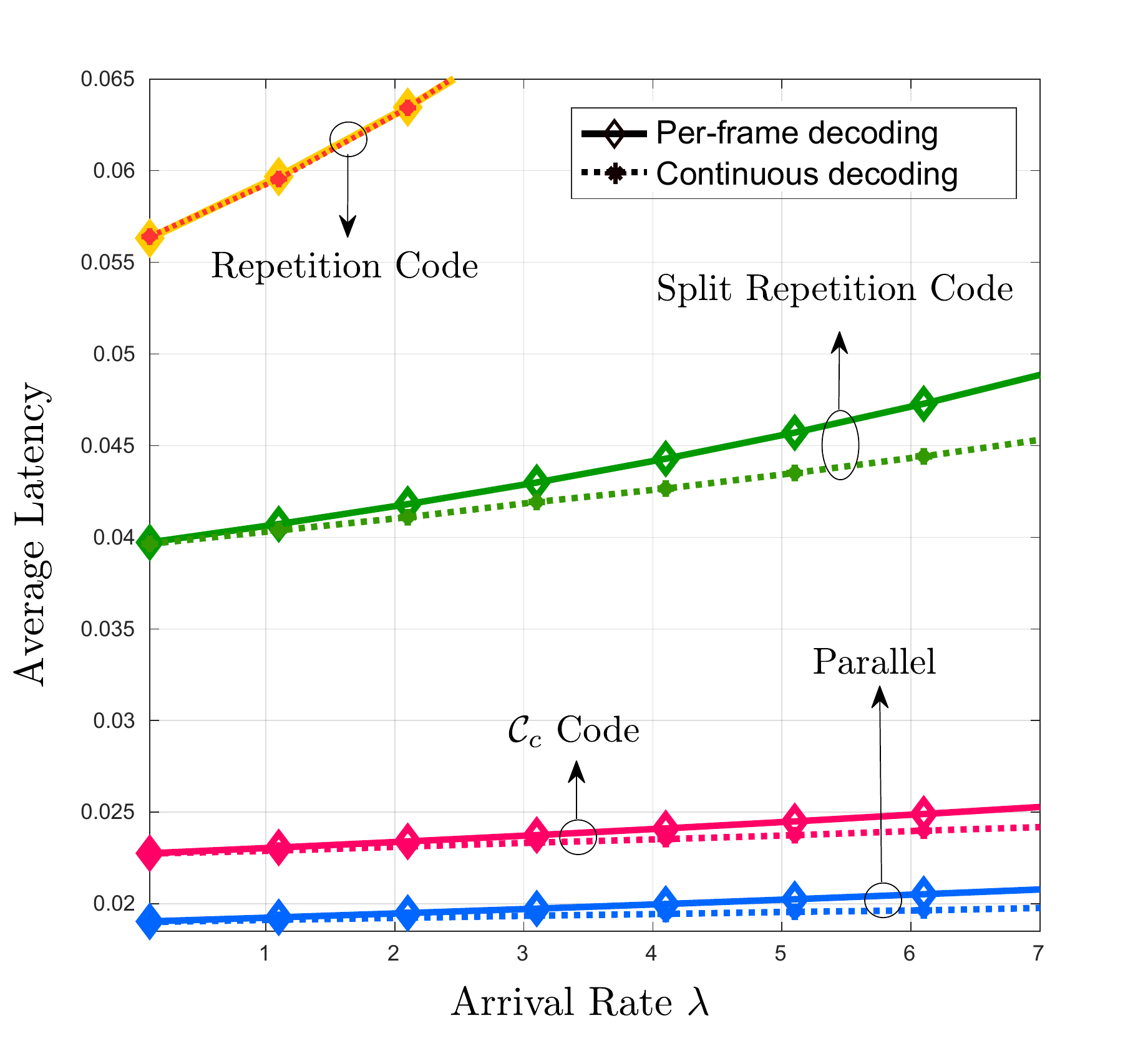}\vspace{-.1cm}~\caption{\footnotesize{ Average latency versus arrival rate $\lambda$ ($L=112, N=8, r=0.5, \mu=500 $). }
   	}~\label{fig111queue}
   	\end{center}\vspace{-1cm}
   \end{figure}
\section{Conclusions}\label{secConclusion}
 In this paper, we analyzed the performance of a novel coded NFV approach for the uplink of a C-RAN system in which decoding takes place at a multi-server cloud processor. The approach is based on the linear combination of the received packets prior to their distribution to the servers or cores, and on the exploitation of the algebraic properties of linear channel codes. The method can be thought of as an application of the emerging principle of coded computing to NFV. 
 {\color{black}In addition, we obtain novel upper bounds on the FUP as a function of the decoding latency based on evaluating tail probabilities for Bernoulli dependent rvs. By extending the analysis from isolated frame transmission to random frame arrival times, the trade-off between average decoding latency and FER for two different policies are derived.}
  Analysis and simulation results demonstrate the {\color{black} benefits} that linear coding of received packets, or NFV coding, can yield in terms of trade-off between decoding latency and reliability. 
 {\color{black} In particular, a prescribed decoding latency or reliability can be obtained by selecting an NFV code with a specific minimum distance and chromatic number, where the two extremes are {\color{black} parallel} NFV-based processing and repetition coding. The former scheme obtains the smallest latency but the lowest reliability, whereas the latter scheme yields the largest latency, but the highest reliability. All other linear NFV codes operate between these two extreme cases.
  
   Among interesting open problems, we mention the design of optimal NFV codes and the extension of the principle of NFV coding to other channels. Note that the approach proposed here applies directly to other additive noise channels in which the user code is an additive group. A key example is the additive Gaussian channel with lattice codes at the user, which will be studied in future work.  }

 \balance

\bibliographystyle{IEEEtran} 
\bibliography{IEEEabrv,references} 

\begin{thebibliography}{10}
\providecommand{\url}[1]{#1}
\csname url@samestyle\endcsname
\providecommand{\newblock}{\relax}
\providecommand{\bibinfo}[2]{#2}
\providecommand{\BIBentrySTDinterwordspacing}{\spaceskip=0pt\relax}
\providecommand{\BIBentryALTinterwordstretchfactor}{4}
\providecommand{\BIBentryALTinterwordspacing}{\spaceskip=\fontdimen2\font plus
\BIBentryALTinterwordstretchfactor\fontdimen3\font minus
  \fontdimen4\font\relax}
\providecommand{\BIBforeignlanguage}[2]{{%
\expandafter\ifx\csname l@#1\endcsname\relax
\typeout{** WARNING: IEEEtran.bst: No hyphenation pattern has been}%
\typeout{** loaded for the language `#1'. Using the pattern for}%
\typeout{** the default language instead.}%
\else
\language=\csname l@#1\endcsname
\fi
#2}}
\providecommand{\BIBdecl}{\relax}
\BIBdecl

\bibitem{aliasgari2017codedisit}
M.~Aliasgari, J.~Kliewer, and O.~Simeone, ``Coded computation against
  straggling decoders for network function virtualization,'' \emph{Proc. IEEE
  International Symposium on Information Theory}, pp. 711--715, Jun., 2018.

\bibitem{mijumbi2016network}
R.~Mijumbi, J.~Serrat, J.-L. Gorricho, N.~Bouten, F.~De~Turck, and R.~Boutaba,
  ``Network function virtualization: State-of-the-art and research
  challenges,'' \emph{IEEE Communications Surveys \& Tutorials}, vol.~18,
  no.~1, pp. 236--262, 2016.

\bibitem{ETSI}
{ \vspace{0mm}European Telecommunications Standards Institute}, ``Network
  function virtualisation {(NFV)}; report on models and features for end-to-end
  reliability,'' Technical Report GS NFV-REL 003, Apr., 2016.

\bibitem{liu2016reliability}
J.~Liu, Z.~Jiang, N.~Kato, O.~Akashi, and A.~Takahara, ``Reliability evaluation
  for {NFV} deployment of future mobile broadband networks,'' \emph{IEEE
  Wireless Communications}, vol.~23, no.~3, pp. 90--96, 2016.

\bibitem{herrera2016resource}
J.~G. Herrera and J.~F. Botero, ``Resource allocation in {NFV}: A comprehensive
  survey,'' \emph{IEEE Transactions on Network and Service Management},
  vol.~13, no.~3, pp. 518--532, 2016.

\bibitem{kang2017trade}
J.~Kang, O.~Simeone, and J.~Kang, ``On the trade-off between computational load
  and reliability for network function virtualization,'' \emph{IEEE
  Communications Letters}, vol.~21, pp. 1767--1770, 2017.

\bibitem{nikaein2015processing}
N.~Nikaein, ``Processing radio access network functions in the cloud: Critical
  issues and modeling,'' in \emph{Proceedings of the 6th International Workshop
  on Mobile Cloud Computing and Services,}.\hskip 1em plus 0.5em minus
  0.4em\relax ACM, Apr., 2015, pp. 36--43.

\bibitem{ETSINFVCRAN}
{ \vspace{0mm}European Telecommunications Standards Institute}, ``Cloud {RAN}
  and {MEC}: A perfect pairing,'' ISBN No. 979-10-92620-17-7, Feb., 2018.

\bibitem{alyafawi2015critical}
I.~Alyafawi, E.~Schiller, T.~Braun, D.~Dimitrova, A.~Gomes, and N.~Nikaein,
  ``Critical issues of centralized and cloudified {LTE-FDD} radio access
  networks,'' in \emph{Communications (ICC), 2015 IEEE International Conference
  on}.\hskip 1em plus 0.5em minus 0.4em\relax IEEE, Jun., 2015, pp. 5523--5528.

\bibitem{nikaein2014openairinterface}
N.~Nikaein, R.~Knopp, F.~Kaltenberger, L.~Gauthier, C.~Bonnet, D.~Nussbaum, and
  R.~Ghaddab, ``Open{A}ir{I}nterface: an open {LTE} network in a {PC},'' in
  \emph{Proceedings of the 20th annual international conference on Mobile
  computing and networking}.\hskip 1em plus 0.5em minus 0.4em\relax ACM, Sep.,
  2014, pp. 305--308.

\bibitem{dotsch2013quantitative}
U.~D{\"o}tsch, M.~Doll, H.-P. Mayer, F.~Schaich, J.~Segel, and P.~Sehier,
  ``Quantitative analysis of split base station processing and determination of
  advantageous architectures for {LTE},'' \emph{Bell Labs Technical Journal},
  vol.~18, no.~1, pp. 105--128, 2013.

\bibitem{rost2014opportunistic}
P.~Rost and A.~Prasad, ``Opportunistic hybrid arq—enabler of
  centralized-{RAN} over nonideal backhaul,'' \emph{IEEE Wireless
  Communications Letters}, vol.~3, no.~5, pp. 481--484, 2014.

\bibitem{khalili2017uplink}
S.~Khalili and O.~Simeone, ``Uplink {HARQ} for cloud {RAN} via separation of
  control and data planes,'' \emph{IEEE Transactions on Vehicular Technology},
  vol.~66, no.~5, pp. 4005--4016, 2017.

\bibitem{Rodriguez17}
V.~Q. Rodriguez and F.~Guillemin, ``Towards the deployment of a fully
  centralized cloud-{RAN} architecture,'' in \emph{Wireless Communications and
  Mobile Computing Conference (IWCMC), 2017 13th International}, Valencia,
  Spain, Jun., 2017, pp. 1055--1060.

\bibitem{rodriguez2018cloud}
------, ``Cloud-ran modeling based on parallel processing,'' \emph{IEEE Journal
  on Selected Areas in Communications}, vol.~36, no.~3, pp. 457--468, 2018.

\bibitem{dean}
J.~Dean and S.~Ghemawat, ``Map{R}educe: simplified data processing on large
  clusters,'' \emph{Commun. of the ACM}, vol.~51, no.~1, pp. 107--113, 2008.

\bibitem{ananthanarayanan2010reining}
G.~Ananthanarayanan, S.~Kandula, A.~G. Greenberg, I.~Stoica, Y.~Lu, B.~Saha,
  and E.~Harris, ``Reining in the outliers in {M}ap-{R}educe clusters using
  mantri.'' in \emph{OSDI}, vol.~10, no.~1, Oct., 2010, p.~24.

\bibitem{zaharia}
M.~Zaharia, M.~Chowdhury, M.~J. Franklin, S.~Shenker, and I.~Stoica, ``Spark:
  Cluster computing with working sets,'' \emph{Proceeding of the 2nd USENIX
  conference on Hot topics in cloud computing}, pp. 10--10, 2010.

\bibitem{li2016coded}
S.~Li, M.~A. Maddah-Ali, and A.~S. Avestimehr, ``Coded distributed computing:
  Straggling servers and multistage dataflows,'' in \emph{Communication,
  Control, and Computing (Allerton), 2016 54th Annual Allerton Conference
  on}.\hskip 1em plus 0.5em minus 0.4em\relax IEEE, Oct., 2016, pp. 164--171.

\bibitem{li2015coded}
------, ``Coded {M}ap{R}educe,'' in \emph{Communication, Control, and Computing
  (Allerton), 2015 53rd Annual Allerton Conference on}.\hskip 1em plus 0.5em
  minus 0.4em\relax IEEE, Oct., 2015, pp. 964--971.

\bibitem{li2016unified}
------, ``A unified coding framework for distributed computing with straggling
  servers,'' in \emph{Globecom Workshops (GC Wkshps), 2016 IEEE}.\hskip 1em
  plus 0.5em minus 0.4em\relax IEEE, Dec., 2016, pp. 1--6.

\bibitem{Ramchandran}
K.~Lee, M.~Lam, R.~Pedarsani, D.~Papailiopoulos, and K.~Ramchandran, ``Speeding
  up distributed machine learning using codes,'' \emph{Proc. IEEE International
  Symposium on Information Theory}, pp. 1143--1147, Jul., 2016.

\bibitem{Li}
S.~Li, M.~A. Maddah-Ali, Q.~Yu, and A.~S. Avestimehr, ``A fundamental tradeoff
  between computation and communication in distributed computing,'' \emph{IEEE
  Transactions on Information Theory}, vol.~64, no.~1, pp. 109--128, 2018.

\bibitem{Yang}
Y.~Yang, P.~Grover, and S.~Kar, ``Computing linear transformations with
  unreliable components,'' \emph{IEEE Transactions on Information Theory},
  2017.

\bibitem{Tandon}
R.~Tandon, Q.~Lei, A.~Dimakis, and N.~Karampatziakis, ``Gradient coding:
  Avoiding stragglers in synchronous gradient descent,'' [Online] www.arxiv.org
  arXiv:1612.03301 [cs.IT], 2016.

\bibitem{Dutta}
S.~Dutta, V.~Cadambe, and P.~Grover, ``Short-dot: Computing large linear
  transforms distributedly using coded short dot products,'' \emph{Advances In
  Neural Information Processing Systems}, pp. 2092--2100, 2016.

\bibitem{Sev17}
A.~Severinson, A.~Graell~i Amat, and E.~Rosnes, ``Block-diagonal coding for
  distributed computing with straggling servers,'' in \emph{Information Theory
  Workshop (ITW)}, Nov., 2017, pp. 464--468.

\bibitem{yu2017polynomial}
Q.~Yu, M.~Maddah-Ali, and S.~Avestimehr, ``Polynomial codes: An optimal design
  for high-dimensional coded matrix multiplication,'' in \emph{Advances in
  Neural Information Processing Systems}, 2017, pp. 4403--4413.

\bibitem{mallick2018rateless}
A.~Mallick, M.~Chaudhari, and G.~Joshi, ``Rateless codes for near-perfect load
  balancing in distributed matrix-vector multiplication,'' \emph{arXiv preprint
  arXiv:1804.10331}, 2018.

\bibitem{kosaian2018learning}
J.~Kosaian, K.~Rashmi, and S.~Venkataraman, ``Learning a code: Machine learning
  for approximate non-linear coded computation,'' \emph{arXiv preprint
  arXiv:1806.01259}, 2018.

\bibitem{wang2015using}
D.~Wang, G.~Joshi, and G.~Wornell, ``Using straggler replication to reduce
  latency in large-scale parallel computing,'' \emph{ACM SIGMETRICS Performance
  Evaluation Review}, vol.~43, no.~3, pp. 7--11, 2015.

\bibitem{joshi2017efficient}
G.~Joshi, E.~Soljanin, and G.~Wornell, ``Efficient redundancy techniques for
  latency reduction in cloud systems,'' \emph{ACM Transactions on Modeling and
  Performance Evaluation of Computing Systems (TOMPECS)}, vol.~2, no.~2, p.~12,
  2017.

\bibitem{ananthanarayanan2013effective}
G.~Ananthanarayanan, A.~Ghodsi, S.~Shenker, and I.~Stoica, ``Effective
  straggler mitigation: Attack of the clones.'' in \emph{NSDI}, vol.~13, Apr.,
  2013, pp. 185--198.

\bibitem{yang2018coded}
Y.~Yang, M.~Chaudhari, P.~Grover, and S.~Kar, ``Coded iterative computing using
  substitute decoding,'' \emph{arXiv preprint arXiv:1805.06046}, 2018.

\bibitem{aktas2017effective}
M.~F. Aktas, P.~Peng, and E.~Soljanin, ``Effective straggler mitigation: Which
  clones should attack and when?'' \emph{ACM SIGMETRICS Performance Evaluation
  Review}, vol.~45, no.~2, pp. 12--14, 2017.

\bibitem{Ali}
A.~Al-Shuwaili, O.~Simeone, J.~Kliewer, and P.~Popovski, ``Coded network
  function virtualization: Fault tolerance via in-network coding,'' \emph{IEEE
  Wireless Communications Letters}, vol.~5, no.~6, pp. 644--647, 2016.

\bibitem{Janson}
S.~Janson, ``Large deviations for sums of partly dependent random variables,''
  \emph{Random Structures \& Algorithms}, vol.~24, no.~3, pp. 234--248, 2004.

\bibitem{polyanskiy}
Y.~Polyanskiy, H.~V. Poor, and S.~Verd{\'u}, ``Channel coding rate in the
  finite blocklength regime,'' \emph{IEEE Transactions on Information Theory},
  vol.~56, no.~5, pp. 2307--2359, 2010.

\bibitem{Amirhossein}
A.~Reisizadehmobarakeh, S.~Prakash, R.~Pedarsani, and S.~Avestimehr, ``Coded
  computation over heterogeneous clusters,'' [Online] www.arxiv.org,
  arXiv:1701.05973 [cs.IT], 2017.

\bibitem{NP}
A.~S{\'a}nchez-Arroyo, ``Determining the total colouring number is {NP}-hard,''
  \emph{Discrete Mathematics}, vol.~78, no.~3, pp. 315--319, 1989.

\bibitem{Brook}
R.~L. Brooks, ``On colouring the nodes of a network,'' \emph{Mathematical
  Proceedings of the Cambridge Philosophical Society}, vol.~37, no.~02, pp.
  194--197, 1941.

\bibitem{Joshi}
G.~Joshi, Y.~Liu, and E.~Soljanin, ``On the delay-storage trade-off in content
  download from coded distributed storage systems,'' \emph{IEEE Journal on
  Selected Areas in Communications}, vol.~32, no.~5, pp. 989--997, 2014.

\bibitem{ross2014introduction}
S.~M. Ross, \emph{Introduction to Probability Models}.\hskip 1em plus 0.5em
  minus 0.4em\relax Academic {P}ress, 2014.

\bibitem{PZK}
H.~C. Tijms, \emph{A {F}irst {C}ourse in {S}tochastic {M}odels}.\hskip 1em plus
  0.5em minus 0.4em\relax John Wiley and Sons, 2003.

\end{thebibliography}

\end{document}